%% file: ITCS_2023_v2.tex
\documentclass[11pt,a4paper]{ouparticle}
\usepackage{graphicx}
\usepackage{tabto}
\usepackage{amsmath}
\usepackage{bbm}
\usepackage{amsthm}
\usepackage{amssymb}
\usepackage{enumerate}
\usepackage{mathtools}
\usepackage{esvect}
\usepackage{caption}
\usepackage{float}
\usepackage{multicol}
\usepackage{tikz}
\usepackage{soul}

\usepackage{hyperref}
\newcommand\myshade{70}
\hypersetup{
	linkcolor  = red!\myshade!black,
	citecolor  = blue!\myshade!black,
	urlcolor   = blue!\myshade!black,
	colorlinks = true,
}

\usepackage{algorithm}
\usepackage[noend]{algorithmic}

\floatname{algorithm}{Scheme}

\usepackage{arydshln}
\usetikzlibrary{positioning}

\usepackage{setspace}
\usepackage{times}
\usepackage{eucal}

\theoremstyle{definition}
\newtheorem{definition}{Definition}
\newtheorem{example}{Example}
\newtheorem{remark}{Remark}

\newtheorem{theorem}{Theorem}
\newtheorem{corollary}[theorem]{Corollary}
\newtheorem{lemma}[theorem]{Lemma}
\newtheorem{proposition}[theorem]{Proposition}

\usepackage{mathtools}

\DeclarePairedDelimiter\abs{\lvert}{\rvert}%
\DeclarePairedDelimiter\norm{\lVert}{\rVert}%

\makeatletter
\let\oldabs\abs
\def\abs{\@ifstar{\oldabs}{\oldabs*}}

\let\oldnorm\norm
\def\norm{\@ifstar{\oldnorm}{\oldnorm*}}
\makeatother

\newcommand{\tr}{{\rm Tr}}
\newcommand{\rs}{{\rm RS}}
\newcommand{\grs}{{\rm GRS}}
\newcommand{\GF}{{\rm GF}}
\newcommand{\rank}{{\rm rank}}

\newcommand{\SSS}{{\color{blue}\mathsf{S}}}
\newcommand{\FF}{\mathbb{F}}
\newcommand{\BB}{\mathbb{B}}
\newcommand{\C}{\mathcal{C}}
\newcommand{\B}{\mathcal{B}}
\newcommand{\A}{\mathcal{A}}

\newcommand{\bM}{\pmb{M}}
\newcommand{\bA}{\pmb{A}}
\newcommand{\bB}{\pmb{B}}
\newcommand{\bR}{\pmb{R}}
\newcommand{\bS}{\pmb{S}}
\newcommand{\bx}{{\pmb{x}}}
\newcommand{\bc}{\pmb{c}}

\newcommand{\todo}[1]{{\color{red}TODO: #1}}

\begin{document}
%\title{\huge Explicit Low-Bandwidth Recovery Schemes for Weighted Sums of Reed-Solomon-Coded Symbols}
\title{\huge Explicit Low-Bandwidth Evaluation Schemes for Weighted Sums of Reed-Solomon-Coded Symbols}
%\title{Low-bandwidth recovery of the linear combination of Reed-Solomon-coded symbols}
%\title{???}
\author{}
\author{Han Mao Kiah, Wilton Kim, Stanislav Kruglik, San Ling, Huaxiong Wang\\[2mm]
\small School of Physical and Mathematical Sciences, Nanyang Technological University, Singapore\\
email: \{hmkiah, wilt0002, stanislav.kruglik, lingsan, hxwang\}@ntu.edu.sg}

% Order to be decided 
\abstract{%
	
	Motivated by applications in distributed storage, distributed computing, and homomorphic secret sharing, we study communication-efficient schemes for computing linear combinations of coded symbols. 
	Specifically, we design low-bandwidth schemes that evaluate the weighted sum of $\ell$ coded symbols in a codeword $\pmb{c}\in\mathbb{F}^n$, when we are given access to $d$ of the remaining components in $\pmb{c}$.
	
	Formally, suppose that $\FF$ is a field extension of $\BB$ of degree $t$. 
	Let $\pmb{c}$ be a codeword in a Reed-Solomon code of dimension $k$ and our task is to compute the weighted sum of $\ell$ coded symbols.
	In this paper, for some $s<t$, we provide an explicit scheme that performs this task by downloading $d(t-s)$ sub-symbols in $\BB$ from $d$ available nodes, whenever $d\ge \ell|\BB|^s-\ell+k$.
	In many cases, our scheme outperforms previous schemes in the literature. 
	
	Furthermore, we provide a characterization of evaluation schemes for general linear codes. 
	Then in the special case of Reed-Solomon codes, we use this characterization to derive a lower bound for the evaluation bandwidth.
}
%	  of some data $\pmb{x}\in\mathbb{F}^k$ given access 
% Motivated by distributed storage, distributed computing, and secret sharing, we consider the setup when each symbol of $\pmb{c}$ is stored by a different entity, and we aim to minimize the total amount of data downloaded from different entities. 

%\todo{
%Our main result is a low-bandwidth schemes to compute the weighted sum of a small number of erased Reed-Solomon coded symbols and lower bounds on the required bandwidth. We highlight that our problem can be solved by schemes for low-bandwidth recovery of several erased symbols in the Reed-Solomon code or low-bandwidth recovery of linear functions of Reed-Solomon encoded data, but in many cases our construction outperforms them. Our construction has applications in distributed storage and coded computation.}
%}
\date{}
\maketitle

\section{Introduction}
\label{sec:intro}

As applications in distributed storage and distributed computation become more prevalent, {\em communication bandwidth} has become a critical performance metric of codes.
In such distributed settings, data is represented as a vector $\bx =\FF^k$ for some finite field $\FF$.
To protect against {\em erasures}, the data is then encoded into a codeword $ \pmb{c}=(c_1,\ldots,c_n)\in\FF^n$ and each codeword symbol $c_i$ is stored in one of $n$ nodes or servers.
The beautiful \textit{Reed-Solomon} (RS) code \cite{originalRSpaper} allows one to recover $\bx$ by accessing any available $k$ nodes. In other words, even if $n-k$ nodes fail, we are able to recover the data. Unfortunately, this solution is not ideal in the scenario where only one node fails. 
In this case, to repair a single node failure, i.e. to recover one codeword symbol, we download $k$ codeword symbols. This is costly as $k$ is usually much greater than one! 
Typically, the total amount of information downloaded is referred to as the \textit{repair bandwidth} and the repair bandwidth problem was first introduced in \cite{Dimakis}. Since then, there has been a flurry of research to design new codes that reduce this repair bandwidth \cite{somesurvey, somesurvey2}.
Of interest to this paper is the pioneering work of Guruswami and Wootters \cite{Guruswami}, where the authors revisit the ubiquitous Reed-Solomon codes and demonstrate that the repair bandwidth can be reduced dramatically when more than $k$ nodes are available. 
Later on, their methods were extended in a variety of scenarios, including multiple erasures and different parameter regimes \cite{Guruswami, DM, Dau, HM1, HM2, SW, Tamo1, Tamo2}.  
\newpage

In this paper, we consider the case where any $\ell$ nodes are unavailable. 
Instead of recovering the contents of these $\ell$ nodes, we recover a {\em weighted sum} or linear combination of them. 
This is motivated by applications in distributed computation.
Specifically, in this setting, we have a computation task ${\cal T}$ that is split into $\ell$ smaller subtasks ${\cal T}_1,{\cal T}_2,\ldots, {\cal T}_\ell$, and we distribute these subtasks to $\ell$ servers to compute locally.
However, if there is a straggler among these $\ell$ servers, that is, a server fails to complete its subtask, we are unable to determine ${\cal T}$. 
To mitigate this {\em straggler} problem, we employ coding and use $n>\ell$ servers.
Specifically, we design a code over finite field $\FF$ and $n-\ell$ additional subtasks ${\cal T}_{\ell+1},{\cal T}_{\ell+2},\ldots, {\cal T}_n$ such that the following holds for any codeword $\bc=(c_1,c_2,\ldots, c_n)$: we have $c_i$ corresponds to the result of subtask ${\cal T}_i$ for $1\le i\le n$. 
As before, we distribute these $n$ codeword symbols to $n$ servers. 
If the code is chosen to be some code that corrects $n-d$ erasures, then we are able to determine $c_1,c_2,\ldots, c_\ell$ as long as we have the results of any $d$ nodes.
Typically, for most computation tasks, the result of ${\cal T}$ is obtained from a weighted sum of $c_1,c_2,\ldots, c_\ell$. 
% (see Example~\ref{ex1} for an application in secure distributed matrix multiplication). 
Hence, the result of ${\cal T}$ is {\em one} symbol in $\FF$.
As before, we see that it is rather wasteful if we simply download $d$ symbols from the $d$ available nodes. Therefore, our central task is to reduce this amount of downloaded information, which we refer to as {\em evaluation bandwidth}.

\subsection{Motivating Example and Some Comparisons}
\label{sec:sdmm}

We illustrate an application scenario of our evaluation schemes. 
Specifically, we look at the problem of {\em communication-efficient secure distributed matrix multiplication} (SDMM), studied by Machado {\em et al.}~\cite{Salim}. 
Consider two matrices $\bA$ and $\bB$ with elements from $\FF$ and we want to compute the product $\bA\bB$ in a {secure} and {distributed} fashion. 
To this end, we split the initial task into $\ell$ subtasks of computing well-defined multiplication of sub-matrices $\bA_i\cdot \bB_i$ ($i\in[\ell]$) where 
$\bA =\left[\begin{array}{c;{2pt/2pt}c;{2pt/2pt}c;{2pt/2pt}c}
	\bA_1 & \bA_2 & \cdots & \bA_\ell
\end{array}\right]$ and 
$\bB^T =\left[\begin{array}{c;{2pt/2pt}c;{2pt/2pt}c;{2pt/2pt}c}
	\bB_1^T & \bB_2^T & \cdots & \bB_\ell^T
\end{array}\right]^T$, and $T$ denotes the transpose of the matrix. Therefore, we want to evaluate $\SSS = \sum_{i=1}^\ell \bA_i\bB_i$.

To both mitigate the presence of stragglers and provide security, Machado {\em et al.} provided the following solution \cite{Salim}. 
First, choose $n+T$ distinct evaluation points $\beta_1,\ldots, \beta_\ell, \beta_{\ell+1}, \ldots, \beta_{\ell+T}, \alpha_1,\ldots,$ $\alpha_{n-\ell}$. 
Next, for some security parameter $T\le n-\ell$, generate $2T$ random matrices $\bR_1,\ldots, \bR_T$ and $\bS_1,\ldots, \bS_T$. Here, $\bR_i$'s and $\bA_i$'s have the same dimensions, while  $\bS_i$'s and $\bB_i$'s have the same dimensions.
Next, we choose two polynomials%
\footnote{Strictly speaking, $\phi(x)$ (and $\psi(x)$) is a multi-valued polynomial. That is, $\phi(x) = \left[\phi_{st}(x)\right]_{st}$ where $\left[\phi_{st}(x)\right]_{st}$ is matrix with the same dimensions as $\bA_i$'s and $\bR_i$'s, and each $\phi_{st}(x)$ is a polynomial of degree at most $\ell+T-1$. Nevertheless, the methods in our work apply with a scaling of the bandwidth.}
$\phi(x)$ and $\psi(x)$ of degrees at most $\ell+T-1$ such that the following hold.
\begin{align*}
	\phi(\beta_i) &= \bA_i & \psi(\beta_i) &= \bB_i & \text{ for }i\in [\ell],\\ \phi(\textcolor{black}{\beta_{\ell+i}}) &= \bR_i & \psi(\textcolor{black}{\beta_{\ell+i}}) &= \bS_i & \text{ for }i\in [T].
\end{align*}
Then for $i\in[n-\ell]$, we provide the server corresponding to $\alpha_i$ the values $\phi(\alpha_i)$ and $\psi(\alpha_i)$, and ask it to compute $\theta(\alpha_i)$ where $\theta(x)=\phi(x)\psi(x)$.

Now, observe that $\theta(x)$ is a polynomial of degree at most $2\ell+2T-2$ (see Section~\ref{sec:prelim} for formal definitions). Hence, we have that $\big(\theta(\beta_1),\ldots, \theta(\beta_\ell), \theta(\alpha_1), \ldots,  \theta(\alpha_{n-\ell})\big)$ belongs to a Reed-Solomon code of length $n$ and dimension $k=2\ell+2T-1$. Hence, we are able to find $\theta(\beta_i)$'s and compute $\SSS$ using the results of any $k$ servers. Furthermore,  Machado {\em et al.} showed that any $T$ servers cannot collude to obtain any information on $\bA_i$'s and $\bB_i$'s~\cite{Salim}. 

To reduce the evaluation bandwidth, Machado {\em et al.} in the same paper then employ \textit{trace polynomials} (see Section~\ref{sec:prelim} for a formal definition)~\cite{Salim}. 
However, as the authors study codes that attain the so-called cut-set bound, the proposed codes are defined over an impractically  large field. 
For example, when $\ell=4$ and $T=36$, then $k=79$, Machado {\em et al.}'s scheme requires a minimal%
\footnote{Here, we refer to the smallest binary field $\GF(2^t)$ where the evaluation bandwidth is less than $kt$, the bandwidth used by the classical approach.} 
binary field $\textrm{GF}(2^{8085})$ and when $\FF=\GF(2^{8085})$, the number of servers $n-\ell$ is at most $128$. 
In this case, when $89$ nodes are available, Machado {\em et al.}'s scheme recovers the four code symbols by downloading $516096$ bits. 
In contrast, Scheme~\ref{scheme:trace} (see Section~\ref{sec:evaluation}) from this work can be defined over $\FF=\GF(2^8)$ and allows up to $n-\ell=252$ servers.
In comparison%
\footnote{Admittedly, relative to the codeword symbol, Machado {\em et al.}'s scheme downloads only 63.83 code symbols, while Scheme~\ref{scheme:trace} downloads 69.5 code symbols and Scheme~\ref{scheme:subspace} downloads 66.9 code symbols. Nevertheless, Scheme~\ref{scheme:trace} and Scheme~\ref{scheme:subspace} perform field operations over a significantly smaller field.}, 
the evaluation bandwidth of Scheme~\ref{scheme:trace} only downloads 556 bits of information whenever 139 nodes are available while the evaluation bandwidth of Scheme~\ref{scheme:subspace} only downloads 535 bits whenever 107 nodes are available.

\begin{figure}[!t]
	\centering
	\includegraphics[width=\textwidth]{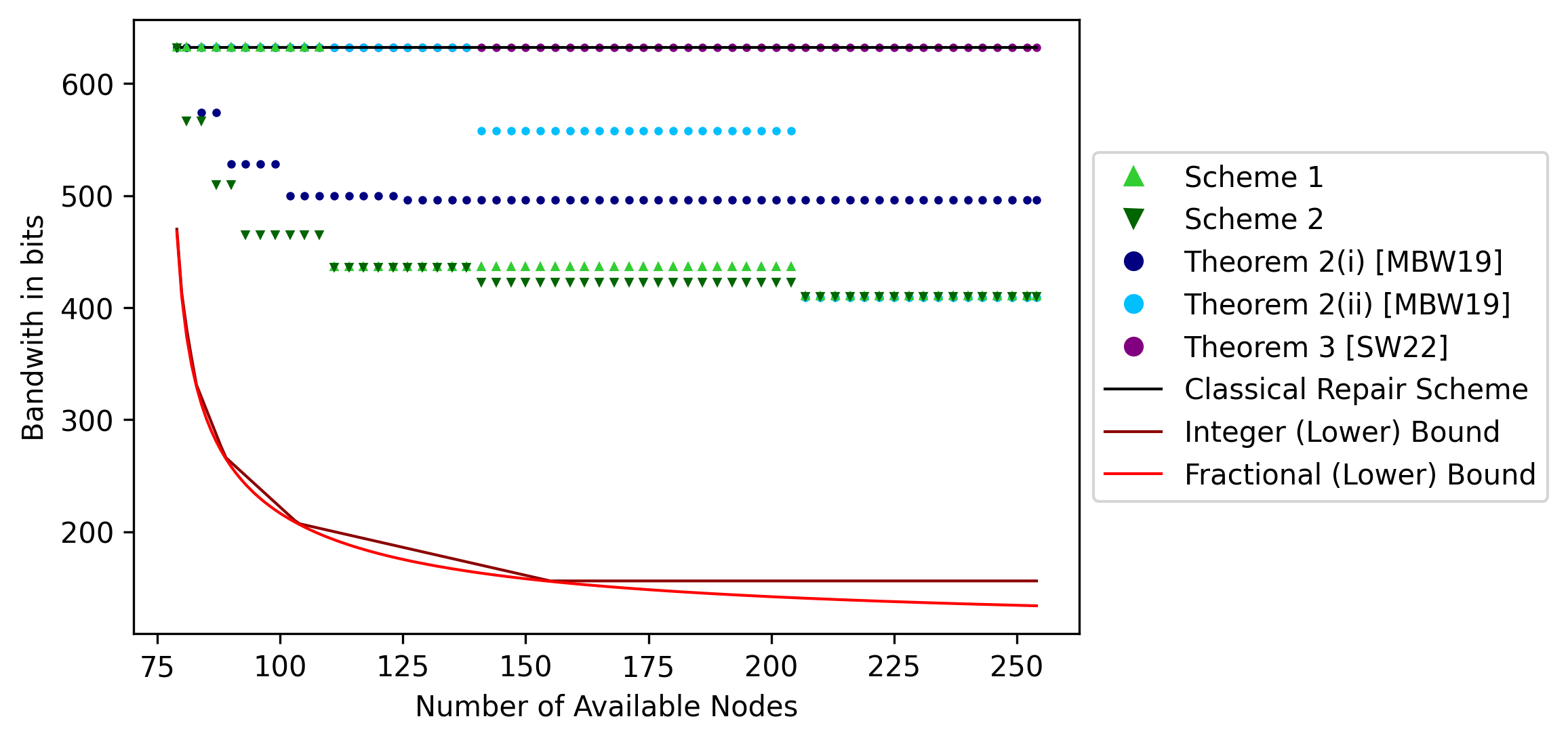}
	\caption{Recovery bandwidth for sums of $\ell=2$ codesymbols of RS code over $\textrm{GF}(2^{8})$ with $k=79$}
	\label{1}
\end{figure}

Next, we briefly describe certain evaluation / repair schemes that apply to our scenario (a detailed discussion is given in both Sections~\ref{sec:related} and~\ref{sec:numerical}). The first work due to Shutty and Wootters~\cite{SW} considers a problem similar to ours.
Specifically, given some data vector $\pmb{x}\in \mathbb{F}^k$ and some linear function $S$, the authors aim to recover $S(\bx)$ from the corresponding codeword of Reed-Solomon code of dimension $k$. In~\cite{SW}, a general low-bandwidth evaluation framework was introduced and the authors demonstrate the existence of low-bandwidth evaluation schemes. %for $S(\bx)$. 
In contrast, for our setup, we are interested in obtaining the weighted sum of $\ell$ coded symbols where $\ell < k$. 
While the schemes in~\cite{SW} clearly apply to our scenario, we exploit the fact that $\ell$ is small to further reduce the bandwidth.
For example, when $\ell=4$, $k=2^{30}$, and $n=|\FF|=2^{32}$, the scheme in~\cite{SW} evaluates $S(\bx)$ by downloading $26\cdot 2^{30}$ bits of information from $n-\ell$ available nodes.
In contrast, for the task with the same set of parameters, Schemes~\ref{scheme:trace} and~\ref{scheme:subspace} only download $8\times 2^{30} $ and $7.5\times 2^{30}$ bits, respectively. We also remark that generally the schemes in~\cite{SW} require large values of $|\FF|$ and $k$ to improve over the classical repair bandwidth.

Another approach to our task is to recover the $\ell$ coded symbols independently and then proceed to evaluate $\SSS$. This is therefore equivalent to the repair problem for {\em multiple erasures}. 
Recently, Dau {\em et al.}~\cite{HM1} generalized the work of Guruswami and Wootters to obtain low-bandwidth repair schemes for two and three erasures.
Later, Mardia, Bartan, and Wootters~\cite{Mardia} extended the repair schemes for any number of erasures and improved the repair bandwidth in some cases.
For example, when $\ell=4$, $k=79$, and $n=|\FF|=2^{8}$, the best scheme in~\cite{Mardia} recovers all four code symbols by downloading $600$ bits of information from $n-\ell$ available nodes. In contrast, when we simply want to find a linear combination of these four symbols, the bandwidth can be reduced. For the same set of parameters, Schemes~\ref{scheme:trace} and~\ref{scheme:subspace} download $556$ and $535$ bits, respectively.

To summarize this example, we fix $n=|\FF|=2^8$ and $k=79$.
In Figure~\ref{1}, we fix $\ell=2$, vary the number of available nodes, and compare the evaluation bandwidths of our schemes with those in~\cite{Mardia,SW}. 
On the other hand, in Table~\ref{table:k79}, we fix the number of available nodes to be $n-\ell$, vary $\ell$ (the size of the weighted sum), and compare the evaluation bandwidths. 
In the table, we also highlight the schemes that yield the lowest bandwidths. In Section~\ref{sec:numerical}, we provide a more comprehensive set of comparisons by varying $d$ and $\ell$. 
Generally, we observe that for most values of $d$ and $\ell$, our schemes outperform those in~\cite{Mardia, SW}.

\begin{table}[t]
\begin{center}
	\small
	\begin{tabular}{|l|r|r|r|r|r|r|}
		\hline
		Scheme                   & $\ell=k=79$ & $\ell=40$ & $\ell=11$ & $\ell=5$ & $\ell=4$ & $\ell=2$ \\ \hline
		Naive                    & $632$    & $632$              & $632$     & $632$     & $632$    & $632$    \\ \hline
		\cite{SW}                & $632$    & $632$              & $632$     & $632$     & $632$    & $632$    \\ \hline
		Scheme~$1$               & $632$    & $632$              & $632$     & $616$     & $556$    & $410$    \\ \hline
		Scheme~$2$               & $632$    & $632$              & \pmb{$630$}     & \pmb{$564$}     & \pmb{$535$}    & $410$    \\ \hline
		\cite[Theorem~2]{Mardia} & $632$    & $632$              & $632$     & $616$     & $600$    & $496$    \\ \hline
		\cite[Theorem~3]{Mardia} & $632$    & $632$              & $632$     & $632$     & $632$    & \pmb{$409$}    \\ \hline
		Fractional (Lower) Bound & $2$ & $64$ & $117$ & $128$ & $130$ & $134$\\ \hline
		Integral (Lower) Bound & $2$ & $80$ & $139$ & $150$ & $152$ & $156$\\ \hline
	\end{tabular}
	\vspace{-5mm}
\end{center}
	\caption{Recovery bandwidth for Reed-Solomon code over $\textrm{GF}(2^8)$ and $k=79$}
	\label{table:k79}
\end{table}

\subsection{Our Contributions}

Therefore, our task can be stated as follows.  Consider a codeword $\pmb{c} = (c_1,c_2,\ldots, c_n)$ of a Reed-Solomon code of length $n$ and dimension $k$.  
Without loss of generality, we consider the first $\ell$ positions and fix $\ell$ nonzero coefficients $\kappa_1,\kappa_2,\ldots, \kappa_\ell \in \mathbb{F}$.
Then given any $d\le n$ nodes, our objective is to design schemes that evaluate the sum $\SSS\triangleq \sum_{i=1}^\ell \kappa_ic_i$ by downloading as little information as possible from these $d$ nodes. 
For ease of exposition, we assume that these $d$ nodes are distinct from the nodes containing $c_1,c_2,\ldots, c_\ell$.

\begin{enumerate}[(A)]
\setlength{\itemsep}{0pt}	
\item {\bf Explicit low-bandwidth evaluation schemes for Reed-Solomon codes}. 
In Section~\ref{sec:evaluation}, we provide two explicit schemes to recover $\SSS=\sum_{i=1}^\ell \kappa_ic_i$ from some $d$ available nodes. We explicitly formulate the information to be downloaded in terms of {\em trace polynomials} and {\em subspace polynomials}. 
 While our methods build upon the frameworks of \cite{Guruswami, HM1, HM2}, a key feature of our schemes is that they can be adjusted according to the parameters: 
 $\ell$ (the number of coded symbols in the sum $\SSS$) and 
 $d$ (the number of available nodes).
 As expected, the evaluation bandwidth of our schemes is smaller when $\ell$ is smaller, or when $d$ is bigger. Specifically, the evaluation bandwidths of the two schemes are as follows.
 
\noindent$\bullet$ {\em Scheme 1}: Let $\FF$ be a field extension of $\BB$ of degree $t$. Whenever $d \ge \ell|\mathbb{B}|^{t-1}-\ell + k$, we can recover $\SSS$ by downloading $d$ subsymbols (in $\BB$), or $d\lceil\log_2|\mathbb{B}|\rceil$ bits, from any $d$ available nodes.

\noindent$\bullet$ {\em Scheme 2}: Suppose further that $W =\{0,w_1,\ldots,w_{|\BB|^s-1}\}$ is a $\BB$-linear subspace of $\FF$ with dimension $s$. Whenever $d \ge \ell|\mathbb{B}|^{s}-\ell + k$, we can recover $\SSS$ by downloading $d(t-s)$ subsymbols (in $\BB$), or $d(t-s)\lceil\log_2|\mathbb{B}|\rceil$ bits, from any $d$ available nodes.

\item {\bf Lower bound for the evaluation bandwidth for Reed-Solomon codes}. By modifying the techniques in \cite{Mardia} and \cite{weiqi}, we provide a characterization of $\BB$-linear evaluation schemes for general linear codes. 
Then in the special case of Reed-Solomon codes, we use this characterization to determine the required evaluation bandwidth. As with our evaluation schemes, the lower bound depends on both $\ell$ and $d$. In the case where $\ell=1$ and $d=n-1$, we recover the bounds obtained in \cite{Guruswami, DM}. Unfortunately, the bound is not tight in general (see Figures~\ref{1} and~\ref{2}, and Tables~\ref{table:k79} and~\ref{table}). 
Nevertheless, this is the first nontrivial $\ell$- and $d$-dependent bound, and is also a lower bound on the repair bandwidth of multiple erasures (studied in \cite{HM1,Mardia}).
%\todo{Nevertheless, it improves the lower bounds in previous work \cite{}. {\em Check that the statement is correct.}}
\end{enumerate}

\section{Preliminaries}
\label{sec:prelim}
Let $[n]$ denote the set of integers $\{1,2,\ldots,n\}$ and let $\FF$ denote a finite field. 
An {\em $\FF$-linear code} $[n, k]$ code $\mathcal{C}$ is an $\FF$-linear subspace of {\em dimension} $k$. Each element of $\cal C$ is called a {\em codeword}. 
The orthogonal complement to $\mathcal{C}$ is called the {\em dual code} $\mathcal{C}^{\perp}$. 
Then by definition, for each $\pmb{c}=(c_1,\ldots, c_n)\in\mathcal{C}$ and $\pmb{c}^{\perp}=(c_1^{\perp},\ldots, c_n^{\perp})\in\mathcal{C}^{\perp}$, we have that $\sum_{i=1}^nc_ic_i^{\perp}=0$. 

Of interest to this paper is the ubiquitous Reed-Solomon codes \cite{originalRSpaper, Dau2022}.

\begin{definition}\label{def:rs}
	Let $\mathbb{F}[x]$ denote the ring of polynomials over $\mathbb{F}$. The {\em Reed-Solomon code} $\rs(\mathcal{W}, k)$ of dimension $k$ with evaluation points $\mathcal{W}= \{\omega_1, \dots , \omega_n\}  \subseteq \mathbb{F}$ and code rate $R = \frac{k}{n}$
	is defined as:
	\[\rs(\mathcal{W}, k) = \left\{
	(f(\omega_1), \ldots , f(\omega_n))
	: f \in \mathbb{F}[x],\, \deg(f) < k\right\}.\]
	On the other hand, the {\em generalized Reed-Solomon code} 
	$\grs(\mathcal{W}, k, \pmb{\lambda})$ of dimension $k$ with evaluation points $\mathcal{W}$ and multiplier $\pmb{\lambda}=(\lambda_1,\ldots,\lambda_n)\in\mathbb{F}^n$, where $\lambda_i\ne0$ for all $i\in [n]$, is defined as:
	\[\grs(\mathcal{W}, k, \pmb{\lambda}) = \left\{
	(\lambda_1f(\omega_1), \ldots , \lambda_nf(\omega_n))
	: f \in \mathbb{F}[x],\, \text{deg}(f) < k\right\}.\]
\end{definition}

Clearly, the generalized Reed-Solomon code $\grs(\mathcal{W}, k, \pmb{\lambda})$ with multiplier vector $\pmb{\lambda}=(1,\ldots,1)$ is a Reed-Solomon code $\textrm{RS}(\mathcal{W}, k)$ with the same set of evaluation points and dimension. Another important fact is that the dual of $\textrm{RS}(\mathcal{W}, k)$ is $\textrm{GRS}(\mathcal{W}, n-k, \pmb{\lambda})$, where the multiplier vector can be explicitly given as (see, for example \cite[Ch. 10, 12]{MacWilliams1977}) 
\begin{equation}\label{GRSM}
	\lambda_i^{-1}=\prod_{j\in [n]\setminus\{i\}}(\omega_i-\omega_j).    
\end{equation}
%Also during our derivations, we employ the fact that a punctured Reed-Solomon code is Reed-Solomon code of smaller length. 
For convenience, when it is clear from context, $f(x)$ denotes a polynomial of degree at most $k-1$ corresponding to a codeword in  $\rs(\mathcal{W},k)$, 
while $r(x)$ denotes a polynomial of degree at most $n-k-1$ corresponding to a dual codeword from $\grs(\mathcal{W}, n-k, \pmb{\lambda})$. 
Our low-bandwidth schemes rely on the existence of low-degree polynomials $r(x)$ and the following relation
\begin{align}\label{PCE}
	\sum_{i=1}^n \lambda_i f(\omega_i) r(\omega_i) = 0.
\end{align}
In what follows, we refer to polynomials $r(x)$ as {\em parity-check polynomials} for $\mathcal{C}$. Of interest, we employ {\em trace polynomials} and {\em subspace polynomials}. Formally, we describe our evaluation framework in the next subsection.

\subsection{Trace Recovery Framework}

Suppose that $\FF$ is a field extension of $\BB$ of degree $t > 1$.
Throughout this paper, we refer to the elements of $\FF$ as {\em symbols} and
the elements of $\BB$ as {\em sub-symbols}. We also denote the nonzero elements of $\FF$ as $\FF_*$. 
As mentioned in the introduction, a remarkable property of Reed-Solomon codes is the following: for any Reed-Solomon code of dimension $k$, we can reconstruct the entire word $\bc$ by accessing any $k$ symbols in $\bc$. 
This is also known as the {\em maximum distance separable (MDS)} property.

Fix $\ell\ge 1$ and some $\ell$-tuple $\pmb{\kappa}\in \mathbb{F}_*^{\ell}$. 
In this work, we are interested in recovering the weighted sum $\SSS=\sum_{i=1}^\ell \kappa_i$ by accessing some $d\ge k$ available nodes. The classical approach downloads $k$ symbols from any $k$-subset of the $d$ available nodes. Over $\BB$, the number of sub-symbols downloaded is $kt$ and the question we want to address is: can we download less than $kt$ sub-symbols, especially when $d$ is much greater than $k$?

Let us formally define a weighted-sum evaluation scheme for general $\FF$-linear code.
\begin{definition}\label{recscheme}
	Let $\mathbb{F}$ be a degree-$t$ field extension of $\mathbb{B}$.
	Let $\ell$ and $d$ be integers and $\mathcal{C}$ be an $\FF$-linear code of length $\ell+d$. 
	We fix a set of coefficients $\pmb{\kappa}\in\FF_*^\ell$.
	A $(\pmb{\kappa},\mathcal{C})$-{\em weighted-sum evaluation scheme} or shortly $(\pmb{\kappa},\mathcal{C})$-{\em evaluation scheme} is defined to be a collection of $d$ functions
	\begin{align}
		g_j\;:\;\mathbb{F}\to\mathbb{B}^{b_j}\text{ for }j\in[d]
	\end{align}
	such that for all codewords $\pmb{c}\in\mathcal{C}$, the weighted sum $\sum_{i=1}^{\ell}\kappa_ic_i$ can be computed using values in the set $\{g_j(c_{\ell+j})\,:\, j\in[d]\}$. 
	The {\em evaluation bandwidth} $b$ is given by $\sum_{j=1}^db_j$. 
	If the functions $g_j (j\in [d])$ along with the function that determines the weighted sum $\sum_{i=1}^{\ell}\kappa_ic_i$  are all $\mathbb{B}$-linear, we say that the $(\pmb{\kappa},\mathcal{C})$-evaluation scheme is $\mathbb{B}$-linear. 
	%In that follows we limit our attention to such schemes only. 
\end{definition}

In this work, we limit our attention to $\BB$-linear evaluation schemes only.
To this end, we consider the {\em trace function} $\tr : \FF\to \BB$ defined as follows:
\begin{equation}
	\tr(x) = \sum_{i=0}^{t-1}x^{|\mathbb{B}|^i} \text{ for } x\in \FF. \label{eq:trace}
\end{equation}
It is straightforward to see that the trace function is $\mathbb{B}$-linear.
Furthermore, $\tr(x)$ is a polynomial in $x$ with degree $|\BB|^{t-1}$. 
Hence, in this paper, we also refer to $\tr(x)$ as a {\em trace polynomial}.
We also refer the interested reader to Lidl and Niedderreiter \cite[Ch. 2]{Lidl} for an in-depth overview of the properties of the trace function. Actually, one crucial fact is that all $\BB$-linear maps defined over $\FF$ can be described in terms of trace functions.

%The $\mathbb{B}$-linearity means that, for any $a, b\in\mathbb{B}$ and $x, y\in\mathbb{F}$, the following equality holds
%\begin{align*}
%	\tr(ax+by)=a\tr(x)+b\tr(y).
%\end{align*}

Thus far, we have a function that maps symbols in $\FF$ to sub-symbols in $\BB$. 
Next, we describe a procedure to recover a symbol in $\FF$ using sub-symbols in $\BB$.
If we view $\mathbb{F}$ as a $\BB$-linear subspace of dimension $t$, we can define a $\BB$-basis $\{u_1,\ldots,u_t\}$ for $\mathbb{F}$. 
Furthermore,  there exists a  {\em trace-dual basis} $\{\widetilde{u}_1,\ldots,\widetilde{u}_t\}$ for $\mathbb{F}$ such that $\tr(u_i\widetilde{u}_j)=1$ if $i=j$, and $\tr(u_i\widetilde{u}_j)=0$, otherwise. The following result plays a crucial role in our evaluation framework.

%\newpage

\begin{proposition}\cite[Ch. 2]{Lidl}\label{prop:trace}
	Let $\{u_1,\ldots,u_t\}$ be a $\BB$-basis of $\mathbb{F}$. Then there exists a trace-dual basis  $\{\widetilde{u}_1,\ldots,\widetilde{u}_t\}$ and we can write each element $x\in\mathbb{F}$ as
	\begin{equation*}
		x=\sum_{i=1}^t\tr(u_ix)\widetilde{u}_i.
	\end{equation*}
\end{proposition}

Hence, following \cite{Guruswami}, our strategy to evaluate $\SSS$ is to compute $t$ independent traces $\tr(u_i\SSS)\,(i\in [t])$ by utilizing as little symbols in $\mathbb{B}$ as possible.

\subsection{Related Work}
\label{sec:related}

As mentioned earlier, our problem of interest bears similarities with certain previous work. 
Namely, one can obtain the evaluation of a weighted sum by employing a scheme that either {\em repairs multiple erasures} or {\em evaluates a weighted sum of all data symbols}.
Hence, in this subsection, we review state-of-the-art schemes for these scenarios and reformulate their results in terms of $(\pmb{\kappa},\C)$-evaluation schemes.

\begin{itemize}
\item \textbf{Low-bandwidth repair of multiple erasures}.
Low-bandwidth repair for Reed-Solomon codes was initiated by the pioneering paper of Guruswami and Wootters \cite{Guruswami}. In this work, Guruswami and Wootters introduced the trace repair framework and designed an optimal repair schemes for the case of single erasure. Specifically, they used trace polynomials as parity-check polynomials and formed equations to determine the traces to download.

Later, in place of trace polynomials, Dau and Milenkovic studied the use of subspace polynomials as parity-check polynomials and designed another class of optimal repair schemes  \cite{DM}. This setup was then generalized to the case of two and three erasures in centralized and decentralized repair models in papers \cite{HM1, HM2,Mardia}. 
The trade-off between the sub-packetization and bandwidth for multiple erasures was investigated in~\cite{weiqi, vardy}. 
Herein, we focus on the case of small sub-packetization level (of the order $\log(n)$).

We also do not impose restrictions on the set of evaluation points and can choose them to be any set of specific size. This is noticeably different from related secret-sharing schemes in which each secret is represented by a single field symbol (see, for example,~\cite{Wang} and the references wherein).
Specifically, in~\cite{Wang}, the evaluation points, which correspond to secrets, are required to satisfy certain algebraic properties.

Now, suppose that we have a low-bandwidth scheme that repairs $\ell$ erasures by downloading information from any $d$ available nodes. 
Then for $\pmb{\kappa}\in \FF_*^{\ell}$, we can first use the repair scheme to recover the coded symbols $c_1,c_2,\ldots, c_\ell$ and then compute the weighted sum $\SSS=\sum_{i\in [\ell]} \kappa_i c_i$.
Therefore, we obtain an $(\pmb{\kappa},\C)$-evaluation scheme with the same bandwidth.
In the following theorem, we summarize the state-of-the-art results for low-bandwidth repair of multiple erasures.

\begin{theorem}[{\cite[Theorems 2 and 3]{Mardia}}]\label{th::MBW}
Let  $\B$ and $\A$ be two disjoint subsets of distinct points in $\mathbb{F}$ with $|\B|=\ell$ and $|\A|=d$ and let $\pmb{\kappa}\in\FF_*^{\ell}$. 
For $\ell \le k\le \ell+d$, let $\C$ be the Reed-Solomon code $\rs(\B\cup\A,k)$. 

Suppose that $\FF$ is a field extension of $\BB$ with degree $t$.
\begin{enumerate}[(i)]
\item If $s\le t$ and $d\ge |\BB|^s(2\ell-1)-2\ell+k+1$, there is an $(\pmb{\kappa},\C)$-evaluation scheme with bandwidth $d(t-s)$ subsymbols (or $d(t-s)\left\lceil\log_2|\BB|\right\rceil$ bits) \cite[Theorem~2]{Mardia}.
\item If $t>\binom{\ell}{2}+\log_{|\BB|}(1+\ell(\ell+\binom{\ell}{2}(|\BB|-1))(|\BB|-1))$ and $d\ge|\BB|^{t-1}-\ell+k$, there is an $(\pmb{\kappa},\C)$-evaluation scheme with bandwidth $d\ell - (|\BB|-1)\binom{\ell}{2}$ subsymbols (or $(d\ell - (|\BB|-1)\binom{\ell}{2})\left\lceil\log_2|\BB|\right\rceil$ bits) \cite[Theorem~3]{Mardia}.
\end{enumerate} 
\end{theorem}

\begin{comment}
\begin{remark}
		For fixed $\ell$, $d$, and $\BB$ and $\ell \ll n^R$, the bandwidth of the scheme in Theorem~\ref{th::MBW}(i) is approximately $\log_{|\BB|}{\frac{2\ell}{1-R}}$ subsymbols per surviving node, while the bandwidth of the scheme in Theorem~\ref{th::MBW}(ii) is approximately $\ell\log_{|\BB|}{\frac{1}{1-R}}$  subsymbols per surviving node. 
\end{remark}
\end{comment}
\begin{remark}
	Of interest are two parameter regimes for the repair of MDS codes that depends on the degree of field extension $t$.
	% that, in fact, reflects the storage on each node. 
	The first one, more commonly studied for regenerating codes, is when $t$ is significantly larger than $n-k$. %This case is more commonly studied for regenerating codes. 
	In this regime, we have the famous cut-set (lower) bound \cite{Dimakis} on repair bandwidth and there are codes attaining it for one or several erasures \cite{Barg}. 
	Another regime, more natural for practical settings, is when $t$ is small compared to $n-k$.
	In this case, the cut-set bound cannot be met, but there exists a matching lower bound for single erasure repair \cite{DM, Guruswami}. We note that deriving matching lower bound for several erasures repair remains open. 
	In our work, we focus on this parameter regime, that is, $t$ is small.
\end{remark}

\item \textbf{Low-bandwidth function evaluation of Reed-Solomon encoded data}.
Recently, Lenz {\em et al.} investigated the problem of designing codes for function evaluation of information symbols transmitted over a noisy channel  \cite{lenz}. In the paper, the authors analyzed the trade-off between channel noise and coding rates for general coding channels. Later, Shutty and Wootters studied such a problem for Reed-Solomon codes and investigated the efficiency on the receiver side \cite{SW}. More formally, they introduced the framework for low-bandwidth recovery of weighted sum of Reed-Solomon encoded data while some of the coded symbols may be erased. 

Now, suppose that we have such a scheme that can tolerate $\ell$ erasures by employing $d$ available nodes. Clearly, for $\pmb{\kappa}\in\FF_*^{\ell}$ the weighted sum $\SSS=\sum_{i\in [\ell]} \kappa_i c_i$ of coded symbols $c_1, c_2,\ldots,c_{\ell}$ can be transformed to the weighted sum of information symbols. Therefore, we obtain an $(\pmb{\kappa},\C)$-evaluation scheme with the same bandwidth. The following theorem formalizes such a result. We defer the explanation on why \cite[Theorem 19]{SW} translates to the form of Theorem~\ref{th::SW} to Appendix~\ref{SWPROOF}.

\begin{theorem}[{\cite{SW}}]\label{th::SW}
Let  $\B$ and $\A$ be two disjoint subsets of distinct points in $\mathbb{F}$ with $|\B|=\ell$ and $|\A|=d$ and let $\pmb{\kappa}\in\FF_*^{\ell}$. 
For $\ell \le k\le \ell+d$, let $\C$ be the Reed-Solomon code $\rs(\B\cup\A,k)$. 

Suppose that $\FF$ is a field extension of $\BB$ with degree $t$.
For any positive number $s$, if $|\mathbb{B}|^{t-1}(s+1)+k+1\leq d\leq |\mathbb{B}|^t-\ell$, then there is a $(\pmb{\kappa},\C)$-evaluation scheme with bandwidth $d|\BB|/s$ subsymbols (or $(d|\BB|/s)\left\lceil\log_2|\BB|\right\rceil$ bits).
\end{theorem}

\begin{comment}
\begin{remark}
For fixed $\ell$, $d$, and $\BB$,  the bandwidth of the scheme in Theorem~\ref{th::SW} is approximately $O\left(\frac{1}{1-R}\right)$ subsymbols per surviving node. 
\end{remark}
\end{comment}

\item {\bf Communication-Efficient Secret Sharing Schemes.} 
	In secret sharing, shares of a secret $S$ is distributed to $d$ participants 
	so that the secret can be recovered with a certain amount $k$ of shares. 
	Secret sharing scheme was first proposed independently by Shamir \cite{shamir} and Blakley \cite{blakleyscheme}. 
	Shamir's scheme~\cite{shamir} is essentially a Reed-Solomon code.
	Here, a secret polynomial of degree at most $k-1$ whose constant term is $S$ is chosen.
	Then the polynomial of degree is evaluated at $d$ distinct nonzero points and distributed to the participants.
	Using polynomial properties, we see that $k$ participants can recover $S$ via interpolation, 
	while any $k-1$ participants cannot obtain any information on $S$.
	In \cite{blakley}, this notion of {\em ramp schemes} was introduced. 
	Here, we have two thresholds are $k$ and $k'$: any $k$ participants can recover the secret $S$, 
	while any subset of at most $k'$ participants obtain no information about $S$.
	Many research has been done to improve ramp scheme, for instance, to construct a communication efficient ramp scheme (see \cite{cess1, cess2,cess3}),
	or to prevent partial information from being recovered explicitly (see \cite{ssss1, ssss2, ssss3,ssss4}). 
	
	Another extension is known as \textit{Homomorphic Secret Sharing}, 
	which aims to evaluate a function $F$ of the secrets $s_1,\ldots,s_N$ from $N$ clients \cite{boyle16,boyle18}. 
	Similar to \cite{SW}, we can employ our schemes for single-client HSS (when $N=1$) with secret $S = \{c_1,\ldots,c_\ell\}$ and $F(\pmb{\kappa},S) = \sum_{i\in[\ell]}\kappa_i c_i$.
	%, where $\pmb{\kappa} = \{\kappa_1,\ldots,\kappa_\ell\}$ is known by the public. 
	Specifically, we can set ${\cal B} = \{\beta_1,\ldots,\beta_\ell\}$ and ${\cal A}=\{\alpha_1,\ldots,\alpha_d\}$ be two disjoint subsets of distinct points in $\FF$. 
	The shares are computed by using a polynomial $f$ of degree $k-1$ satisfying $f(\beta_i) = c_i$, for all $i\in[\ell]$, 
	and we distribute $f(\alpha_1),\ldots,f(\alpha_d)$ to the $d$ participants. 
	Hence, the secret together with the shares form a codeword of $\rs({\cal B}\cup{\cal A},k)$ and we can apply our result to compute $F(\pmb{\kappa},s)$. The important difference is that in HSS, the location of codesymbols for which we want to find the weighted sum is fixed, but our results have no such restriction.

\end{itemize}

\subsection{Results Overview}

As mentioned earlier, our first result is the explicit construction of two low-bandwidth weighted sum evaluation schemes. 
The schemes are formally described in Section~\ref{sec:evaluation} and their properties are summarized in the following theorem.
\begin{theorem}\label{thm:scheme}
Let  $\B$ and $\A$ be two disjoint subsets of distinct points in $\mathbb{F}$ with $|\B|=\ell$ and $|\A|=d$ and let $\pmb{\kappa}\in\FF_*^{\ell}$. 
For $\ell \le k\le \ell+d$, let $\C$ be the Reed-Solomon code $\rs(\B\cup\A,k)$. 
\begin{enumerate}[(i)]
\item Suppose $\FF$ is a field extension of $\BB$ with degree $t$. 
If $d\geq \ell|\mathbb{B}|^{t-1}-\ell+k$, then Scheme~\ref{scheme:trace} is a $(\pmb{\kappa},\mathcal{C})$-weighted sum evaluation scheme with bandwidth $d$ subsymbols (or $d\left\lceil \log_2|\BB|\right\rceil$ bits).
\item Suppose further $W$ is a $\BB$-linear subspace of $\FF$ with dimension $s$.
If $d\geq \ell|\mathbb{B}|^{s}-\ell+k$, then Scheme~\ref{scheme:subspace} is a  $(\pmb{\kappa},\mathcal{C})$-weighted sum evaluation scheme with bandwidth $d(t-s)$ subsymbols (or $d(t-s)\left\lceil \log_2|\BB|\right\rceil$ bits).	
%In other words, to evaluate the weighted sum, we download sub-symbols in $\mathbb{B}$ or $ d\lceil\log_2|\mathbb{B}|\rceil$ bits.
\end{enumerate}
\end{theorem}
\begin{remark}
    We provide some asymptotic estimates of the bandwidth for the scheme in Theorem~\ref{thm:scheme} and make comparisons with the schemes in Theorems~\ref{th::MBW}~and~\ref{th::SW}.
	In Section~\ref{sec:numerical}, we also provide numerical values of the bandwidths for certain fixed parameters. \textcolor{black}{Let us fix the base field $\BB$ and let the codelength $n$ grow for some constant code rate $k=Rn$.}
	
	\textcolor{black}{First,} let $\ell = o(n^R)$. Then the bandwidth incurred by the schemes in Theorem~\ref{thm:scheme} is approximately $\log_{|\BB|} \frac{\ell}{1-R}$ subsymbols for each available node.
	In contrast, the bandwidth of the scheme in Theorem~\ref{th::MBW}(i) is approximately $\log_{|\BB|}{\frac{2\ell}{1-R}}$ subsymbols for each available node, while the bandwidth of the scheme in Theorem~\ref{th::MBW}(ii) is approximately $\ell\log_{|\BB|}{\frac{1}{1-R}}$  subsymbols for each available node. In other words, the schemes in our work incurs slightly less bandwidth than the scheme in Theorem~\ref{th::MBW}(i), and 
	saves a factor of $\ell$ as compared to the scheme in Theorem~\ref{th::MBW}(ii). In contrast, the scheme in Theorem~\ref{th::SW} necessarily incurs ${\frac{C}{1-R}}$ subsymbols for each available node for some constant $C\ge 1$. Therefore, \textcolor{black}{in this regime, our schemes provide significantly better bandwidth values.} 
	
	\textcolor{black}{Next, we consider another regime with $\ell=\Theta(n)$. Then the total bandwidth of schemes in Theorem~\ref{th::MBW}(i), (ii) and  Theorem~\ref{thm:scheme} becomes equal to $kt$ and thus trivial. In contrast, the scheme in Theorem~\ref{th::SW} incurs ${\frac{C_1}{C_2-R}}$ subsymbols for each available node for some non-negative constants $C_1\ge 1$ and $C_2<1$ and the total bandwidth can be less than $kt$. Therefore, in this regime, our schemes lose to scheme in Theorem~\ref{th::SW}.}
\end{remark}

\begin{comment}
\begin{remark}
		For fixed $\ell$, $d$, and $\BB$,  the bandwidth of the scheme in Theorem~\ref{thm:scheme} is approximately $\log_{|\BB|} \frac{\ell}{1-R}$ subsymbols per surviving node. 
\end{remark}
\end{comment}
Our next result is a lower bound on the evaluation bandwidth of weighted-sum schemes involving Reed-Solomon codes.
\begin{theorem}\label{thm:lowerbound}
Let  $\B$ and $\A$ be two disjoint subsets of distinct points in $\mathbb{F}$ with $|\B|=\ell$ and $|\A|=d$. For $\ell \le k\le \ell+d$, let $\C$ be the Reed-Solomon code $\rs(\B\cup\A,k)$.
Then for any $\pmb{\kappa}\in\FF_*^{\ell}$, we have that the evaluation bandwidth is at least $b_{\min}$ subsymbols where 
\begin{equation}\label{eq:bmin}
	b_{\text{min}} = n_0 \left\lfloor \log_{|\mathbb{B}|} \frac{d}{L}\right\rfloor + (d-n_0)  \left\lceil \log_{|\mathbb{B}|} \frac{d}{L} \right\rceil,
\end{equation}
and 
\begin{align}
	n_0 & = \left\lfloor \frac{L-d|\mathbb{B}|^{-\left\lceil \log_{|\mathbb{B}|} \frac{d}{L} \right\rceil}}{|\mathbb{B}|^{-\left\lfloor \log_{|\mathbb{B}|} \frac{d}{L}\right\rfloor} - |\mathbb{B}|^{-\left\lceil \log_{|\mathbb{B}|} \frac{d}{L}\right\rceil}}\right\rfloor\,, \label{eq:n0}\\
	L & =\frac{1}{|\mathbb{F}|}\Big((|\mathbb{F}|-1)(\ell+d-k-1)+ d\Big)\,. \label{eq:L}
\end{align}
\end{theorem}

It is clear that any $\ell\le k$ code symbols of the Reed-Solomon code are linearly independent as vectors over $\BB$, hence we cannot recover $\ell$ erased code symbols with bandwidth smaller than those for evaluation of their weighted sum. As a result, we can formulate the following corollary. 

\begin{corollary}
Let  $\B$ and $\A$ be two disjoint subsets of distinct points in $\mathbb{F}$ with $|\B|=\ell$ and $|\A|=d$. For $\ell \le k\le \ell+d$, let $\C$ be the Reed-Solomon code $\rs(\B\cup\A,k)$.
Then for erased symbols $c_1, c_2,\ldots, c_{\ell}$ we have that the recovery bandwidth is at least $b_{\min}$ subsymbols, as defined by~\eqref{eq:bmin},~\eqref{eq:n0} and~\eqref{eq:L}.
\begin{comment}
\begin{equation}\label{eq:bmin-2}
	b_{\text{min}} = n_0 \left\lfloor \log_{|\mathbb{B}|} \frac{d}{L}\right\rfloor + (d-n_0)  \left\lceil \log_{|\mathbb{B}|} \frac{d}{L} \right\rceil,
\end{equation}
and 
\begin{align}
	n_0 & = \left\lfloor \frac{L-d|\mathbb{B}|^{-\left\lceil \log_{|\mathbb{B}|} \frac{d}{L} \right\rceil}}{|\mathbb{B}|^{-\left\lfloor \log_{|\mathbb{B}|} \frac{d}{L}\right\rfloor} - |\mathbb{B}|^{-\left\lceil \log_{|\mathbb{B}|} \frac{d}{L}\right\rceil}}\right\rfloor\,, \label{eq:n0-2}\\
	L & =\frac{1}{|\mathbb{F}|}\Big((|\mathbb{F}|-1)(\ell+d-k-1)+ d\Big)\,. \label{eq:L-2}
\end{align}
\end{comment}
\end{corollary}
\begin{remark}\label{rem:l=1}
\textcolor{black}{It is clear that for $\ell=1$, results of Theorems~\ref{thm:scheme} and \ref{thm:lowerbound} coincide with results for one coded symbol recovery from~\cite{Guruswami, DM}.}
\end{remark}
\begin{comment}
\todo{Within this paper, we consider the problem of efficient evaluation of the weighted sum of Reed-Solomon coded symbols. First, we propose two low-bandwidth schemes based on forming a codeword of dual code with certain properties. To do so, we employ trace-mapping functions and subspace polynomials. We explicitly state the corresponding evaluation schemes and provide the upper bound on the bandwidth. Second, we characterize the evaluation scheme for the weighted sum of \textit{any} code that is linear over $\mathbb{F}$ in terms of special algebraic structure. After it, we employ it to prove the lower bound on evaluation bandwidth for the Reed-Solomon code, and also the lower bound on the repair bandwidth for several erasures in it.  The behavior of proposed schemes and comparison with competing constructions and lower bounds are given in the numerical results section. Finally, we conclude the paper and give some possible research directions.}
\end{comment}
\vspace{-5mm}

\section{Low-Bandwidth Evaluation Schemes}
\label{sec:evaluation}
\input{lowbandwidth.tex}

\section{Lower Bound on Evaluation Bandwidth}
\label{sec:lowerbound}

%\subsection{Repair matrix framework}
In this section, we provide a detailed proof of Theorem~\ref{thm:lowerbound}.
First, following~\cite{Mardia, weiqi}, we characterize a $(\pmb{\kappa},\mathcal{C})$-weighted sum recovery scheme for \textit{any $\FF$-linear code} in terms of a matrix whose columns belong to its dual code. 
%the dual code by means of additional algebraic structure. 

Henceforth, $\FF$ is an extension field of $\BB$ with degree $t$. 
Hence, each element in $\mathbb{F}$ can be considered as a column vector of length $t$ over $\mathbb{B}$. In particular, we represent a vector in $\FF^d$ 
as a $(t\times d)$-matrix $\bM$ over $\mathbb{\BB}$. 
Then we use $\rank_{\mathbb{B}}(\bM)$ to denote the rank of $\bM$ over $\BB$.
We also adopt the following matrix-vector notation. 
For a matrix $\pmb{M}$, we use $\pmb{M}[i,j]$ to denote the entry in the $i$-th row and $j$-th column.
Moreover, we use $\pmb{M}[i,:]$ and $\pmb{M}[:,j]$ to refer to the $i$-th row and $j$-th column of $\bM$, respectively.
% to refer to all elements in the $j$-th column. In that follows we employ bold symbols to denote matrices over $\mathbb{B}$ and vectors over $\mathbb{F}$.

\begin{definition}\label{def:evalmatrix}
Let $\mathcal{C}\subseteq \mathbb{F}^{\ell + d}$ be an $\FF$-linear code. 
Let us fix $\pmb{\kappa}\in\mathbb{F}_*^{\ell}$. 
Then an $(\ell+d)\times t$ matrix $\pmb{M}$ whose entries belong to $\mathbb{F}$ is called a $(\pmb{\kappa},\mathcal{C})$-{\em evaluation matrix} with bandwidth $b$ subsymbols if the following conditions hold:
\begin{enumerate}[(C1)]
\setlength{\itemsep}{0pt}	
\item Columns of $\pmb{M}$ are codewords in the dual code $\mathcal{C}^{\perp}$.
\item For $i\in[\ell]$, the following holds 
\begin{equation}\label{eq:matrix-top}
	\left(\prod_{i\ne1}\kappa_i\right)\pmb{M}[1,:]
	=\left(\prod_{i\ne2}\kappa_i\right)\pmb{M}[2,:]=\cdots
	=\left(\prod_{i\ne\ell}\kappa_i\right)\pmb{M}[\ell,:].
\end{equation}
For convenience, we set $\pmb{M}_0=(\prod_{i\ne1}\kappa_i)\pmb{M}[1,:]$. 
% be the common matrix. 
Furthermore, we have $\rank_{\mathbb{B}}\pmb{M}_0=t$. 
\item We have that 
\begin{equation}\label{eq:matrix-bandwidth}
	\sum_{j\in[d]}\rank_{\mathbb{B}}\pmb{M}[\ell+j,:]=b.    
\end{equation} 
\end{enumerate}
\end{definition}

\begin{comment}
Specifically, for codeword $\pmb{c}=(c_1,\ldotsc_{\ell+d})\in\mathcal{C}$, the recovery matrix $\pmb{M}$ determines $t$ parity check equations over $\mathbb{F}$. More formally, for all $i\in[t]$ we have
\begin{align}
	\pmb{M}[1,i]c_{1}+\cdots+\pmb{M}[\ell,i]c_{\ell}+\pmb{M}[\ell+1,i]c_{\ell+1}+\cdots+\pmb{M}[\ell+d,i]c_{\ell+d}=0.
\end{align}
\end{comment}

Using this notion of evaluation matrix, we characterize when a code $\C$ admits an evaluation scheme with certain bandwidth.

\begin{theorem}\label{thm:characterization}
	Let $\mathcal{C}\subseteq \mathbb{F}^{\ell + d}$ be an $\FF$-linear code and
	$\pmb{\kappa}\in \mathbb{F}_*^{\ell}$. 
	Then $\mathcal{C}$ admits a $(\pmb{\kappa},\mathcal{C})$-weighted-sum evaluation scheme with bandwidth $b$ subsymbols if and only if there exists a $(\pmb{\kappa},\mathcal{C})$-evaluation matrix with bandwidth $b$ that satisfies the conditions in Definition~\ref{def:evalmatrix}.
\end{theorem}

The proof of Theorem~\ref{thm:characterization} is similar to that in~\cite{Mardia, weiqi} and so, we defer it to the appendix.
Instead, we construct the evaluation matrix corresponding to Scheme~\ref{scheme:trace} and verify Conditions (C1), (C2) and (C3).

\begin{example}
Let $\B$ and $\A$ be sets of evaluation points as defined in Section~\ref{sec:evaluation}.
Set $\C$ to be the Reed-Solomon code $\rs(\mathcal{B}\cup\mathcal{A},k)$ and we consider the $(\pmb{\kappa},\mathcal{C})$-weighted-sum evaluation scheme given by Scheme~\ref{scheme:trace}.
Let $\{u_1,u_2,\ldots, u_t\}$ be $\BB$-basis of $\FF$, then the corresponding evaluation matrix $\pmb{M}$ is
\begin{equation*} 
	\pmb{M} = \left[ \begin{matrix} 
		\kappa_1u_1 & \cdots & \kappa_1u_t \\ 
		\vdots & \ddots & \vdots \\ 
		\kappa_{\ell}u_1 & \ldots & \kappa_{\ell}u_t \\ \dfrac{\sigma_{1,1}\lambda_{\ell+1}g(\alpha_1)}{\prod_{j\in[\ell]}(\alpha_1-\beta_j)} & \cdots & \dfrac{\sigma_{t,1}\lambda_{\ell+1}g(\alpha_1)}{\prod_{j\in[\ell]}(\alpha_1-\beta_j)} \\ 
		\vdots & \ddots & \vdots \\ \dfrac{\sigma_{1,d}\lambda_{\ell+d}g(\alpha_d)}{\prod_{j\in[\ell]}(\alpha_d-\beta_j)} & \ldots & \dfrac{\sigma_{t,d}\lambda_{\ell+d}g(\alpha_d)}{\prod_{j\in[\ell]}(\alpha_d-\beta_j)} \end{matrix} \right]\, ,
\end{equation*}
where $g(x)$ and $\sigma_{i,j}$ are defined in Scheme~\ref{scheme:trace}.

To verify Condition (C1), we recall the definition of $r_m(x)$ in the proof of Theorem~\ref{thm:scheme}(i) for $m\in[t]$.
Then we observe that $\kappa_i u_m =\lambda_i r_m(\beta_i)$ for $i\in [\ell]$, and that 
$\frac{\sigma_{m,i} \lambda_{\ell+i} g(\alpha_i)}{\prod_{j\in[\ell]}(\alpha_1-\beta_j)}=\lambda_{\ell+i}r_m(\alpha_i)$.
In other words, the $m$-th column is given by $(\lambda_1 r_m(\beta_1),\ldots, \lambda_\ell r_m(\beta_\ell), \lambda_{\ell+d} r_m(\alpha_{\ell+1}),\ldots, \lambda_{\ell+d} r_m(\alpha_{\ell+d}))^T$. 
Since $r_m(x)$ is a polynomial of degree at most $\ell+d-k-1$, the column is a codeword in the dual of $\rs(\B\cup\A,k)$.

To verify Condition (C2), we observe that $\left(\prod_{j\ne i} \kappa_j\right) \bM[i,:] = \left(\prod_{j\ne i} \kappa_j\right)[u_1,u_2,\ldots, u_t]$. Since $\{u_1,u_2,\ldots, u_t\}$ is a basis, we have that $\left(\prod_{j\ne i} \kappa_j\right)[u_1,u_2,\ldots, u_t]$ is of full-rank over $\BB$, as required by Condition (C2).

Finally, for $i\in [d]$, we have that 
\[\bM[\ell+i,:] = \Big[\sigma_{1,i}, \sigma_{2,i}, \ldots, \sigma_{t,i}\Big] \frac{\lambda_{\ell+i}g(\alpha_i)}{\prod_{i\in[\ell]}(\alpha_i-\beta_j)}\,,  \] 
and thus, $\rank_{\BB}\bM[\ell+i,:]=1$. Therefore, we verify Condition (C3) and have that the corresponding evaluation bandwidth is $d$ subsymbols, corroborating Theorem~\ref{thm:scheme}(i).
\end{example}
\begin{comment}
The column $i\in[t]$ is formed by evaluations of the parity-check polynomial $r_i(x)$ \Big(defined in the proof of Theorem~\ref{thm:scheme}(i)\Big) at the points $\mathcal{B}\cup\mathcal{A}$ multiplied by coefficients $\pmb{\lambda}$. % of the Generalized Reed-Solomon code defined by equations~\ref{GRSM}. 
Utilizing the notation from the description of Scheme~\ref{scheme:trace} and the fact that for basis $\{u_1,\ldotsu_t\}$ of $\mathbb{F}$ over $\mathbb{B}$, we have $r_i(\beta_i)=\dfrac{\kappa_i}{\lambda_i}u_i$. 
\end{comment}

%\subsection{Lower bound on repair bandwidth for Reed-Solomon code}

\vspace{2mm}

As before, we pick two disjoint subsets $\mathcal{B}\triangleq\{\beta_i : i\in [\ell]\}$, $\mathcal{A}\triangleq\{\alpha_j: j\in [d] \}$ and $k\le \ell+d$.
We set $\C$ to be the Reed-Solomon code $\rs(\B\cup\A,k)$ and for the rest of the section, we derive a lower bound for the evaluation bandwidth of a $(\pmb{\kappa},\mathcal{C})$-weighted sum recovery scheme.

First, we characterize the evaluation matrix (defined in Theorem~\ref{thm:characterization}) in terms of parity-check polynomials in its dual code. Specifically, we have the following lemma.

\begin{lemma}\label{lem:poly}
Let $\mathcal{C}$ be $\rs(\B\cup\A,k)$.
If $\pmb{M}\in\mathbb{F}^{(\ell+d)\times t}$ is a $(\pmb{\kappa},\mathcal{C})$-evaluation matrix, 
then there exist $t$ polynomials $p_1(X),\ldots, p_t(X)$ with $\deg p_j(X)\leq \ell +d-k-1$ so that for all $j\in[t]$, 
\begin{equation}
	\bM[i,j] = 
	\begin{cases}
		\lambda_ip_j(\beta_i),            & \text{ if } i \le \ell,\\
		\lambda_ip_j(\alpha_{i-\ell}),    & \text{ if } \ell+1\le i \le d+\ell.
	\end{cases}
\end{equation}
Here, $\lambda_i$ are multipliers of the dual GRS code defined by~\eqref{GRSM}. 
\end{lemma}

\begin{proof}
The result immediately follows from the fact that the dual of the Reed-Solomon code is a generalized Reed-Solomon code with multipliers given by~\eqref{GRSM} and Condition (C1) of an evaluation matrix definition.
\end{proof}

Let ${\cal P}$ be the set of polynomials $p_1(X),\ldots,p_t(X)$ in Lemma~\ref{lem:poly}.
For $\bx=(x_1,\ldots,x_t)\in\mathbb{B}^t$, we set $p_{\bx}(X)\triangleq x_1p_1(X) + \cdots + x_tp_t(X)$, a $\mathbb{B}$-linear combination of polynomials in ${\cal P}$. Then for $\omega\in\B\cup \A$, we consider the set $S_{\omega}=\{\bx\in\mathbb{B}^t\setminus  \{ \pmb{0} \} :\, p_\bx(\omega) = 0\}$ and the quantity
\begin{equation}
	\Omega=\frac{1}{|\FF|-1}\sum_{\bx\in \mathbb{B}^t\setminus\{\pmb{0}\}}\sum_{\omega\in \mathcal{B}\cup\mathcal{A}}\mathbb{I}(\bx\in S_{\omega}).
\end{equation}
Here, $\mathbb{I}$ is the indicator function. 

Now, for $i\in[d]$, we consider the value $b_i\triangleq \rank_{\BB}(\pmb{M}[\ell+i,:])=$ $\rank_{\mathbb{B}}\{\lambda_{i}p_1(\alpha_i),\ldots,\lambda_{i}p_{t}(\alpha_i)\}=\rank_\mathbb{B}\{p_1(\alpha_i),\ldots,p_t(\alpha_i)\}$.
Then from Condition~(C3) of Definition~\ref{def:evalmatrix}, we have that $\sum_{i\in[d]}b_i=b$.
In what follows, we compute $\Omega$ in two different ways and obtain certain constraints on the values of $b_i$'s.
\begin{lemma}\label{lem:Omega-1}
The quantity $\Omega$ is given by 
	\begin{equation*}
		\Omega=\frac{1}{|\FF|-1}\sum_{i\in[d]}(|\BB|^{t-b_{i}}-1).
	\end{equation*}
\end{lemma}

\begin{proof}
First, observe that $\bx\in S_{\omega}$ if and only if $\bx$ belongs to the null space of the matrix $[p_1(\omega),\ldots,p_t(\omega)]$.
In other words, for $\omega\in \B\cup \A$, we have $\sum_{\bx\in \mathbb{B}^t\setminus\{\pmb{0}\}}\mathbb{I}(\bx\in S_{\omega}) = |\BB|^{t-b_\omega}-1$,
where $b_\omega$ is the rank of $[p_1(\omega),\ldots,p_t(\omega)]$. 
Now, if $\omega\in \B$, we have that $b_\omega=t$ from Condition (C2) and so, $\sum_{\bx\in \mathbb{B}^t\setminus\{\pmb{0}\}}\mathbb{I}(\bx\in S_{\omega})=0$. 
On the other hand, for $i\in [d]$, we have that $b_{\alpha_i}=b_i$ and so,  $\sum_{\bx\in \mathbb{B}^t\setminus\{\pmb{0}\}}\mathbb{I}(\bx\in S_{\alpha_i})=|\BB|^{t-b_i}-1$. Therefore, after switching the summation order, we have
\[\sum_{\bx\in \mathbb{B}^t\setminus\{\pmb{0}\}}\sum_{\omega\in \mathcal{B}\cup\mathcal{A}}\mathbb{I}(\bx\in S_{\omega}) 
= \sum_{\omega\in \mathcal{B}\cup\mathcal{A}} \sum_{\bx\in \mathbb{B}^t\setminus\{\pmb{0}\}}\mathbb{I}(\bx\in S_{\omega})
= \sum_{i\in[d]}(|\BB|^{t-b_{i}}-1)\,.
\] 
Hence, the lemma follows.
\end{proof}

\begin{lemma}\label{lem:Omega-2}
We have the following inequality:
\begin{equation*}
	\Omega\leq(l+d)-k-1.
\end{equation*}
\end{lemma}
\begin{proof}
Now, $\bx^{*}\in S_\omega$ implies that $\omega$ is a root of the polynomial $x_1^*p_1(X)+\cdots+x_t^*p_t(X)$.
Since $\Omega$ is the average number of such roots, there exists  $\bx^*\in\mathbb{B}^t\setminus\{\pmb{0}\}$ 
so such polynomial $x_1^*p_1(X)+\cdots+x_t^*p_t(X)$ has at least $\Omega$ roots. 
On the other hand, since the polynomials $p_1(X),\ldots,p_t(X)$ have degree at most $\ell+d-k-1$, the polynomial $x_1^*p_1(X)+\cdots+x_t^*p_t(X)$ has no more than $l+d-k-1$ roots and so, the lemma follows.
\end{proof}

Putting Lemmas~\ref{lem:Omega-1} and~\ref{lem:Omega-2} together, we have that $\sum_{i\in\mathcal{A}}b_i\geq b_{\min}$, where
\begin{equation}\label{optimization}
	b_{\min}=\min_{b_{i}\in\{0,1,\ldots,t\}}\left\{\sum_{i\in [d]}b_{i}: \sum_{i\in [d]} |\mathbb{B}|^{-b_{i}}\leq L\right\}\,.
\end{equation}
Here,  $L$ denotes the quantity $\frac{1}{|\mathbb{F}|}\Big[(|\mathbb{F}|-1)(\ell+d-k-1)+ d\Big]$ (see also, \eqref{eq:L}).

If we relax $b_i$'s to take fractional values, the minimum is achieved when $b_1=b_2=\ldots = b_d$ and we have the optimal value
\begin{equation}\label{bound:fractional}
	b_{\min}\ge b^*_{\min}=d\log_{|\mathbb{B}|}\frac{|\mathbb{F}|d}{(|\mathbb{F}|-1)(\ell + d-k-1)+ d}\,.
\end{equation}

As in \cite{DM}, the exact solution can also be obtained by solving the integer program~\eqref{optimization}. 
That is, we insist that $b_{i}\in\{0,1,\ldots,t\}$ for all $i\in [d]$.
More precisely, for the case where $\log_{|\mathbb{B}|} \frac{d}{L} \notin \mathbb{Z}$, the following bound improves \eqref{bound:fractional}.
Specifically, we let $n_0$ be as defined in \eqref{eq:n0} and set 
\begin{equation}
	b_1^* = \cdots = b_{n_0}^* = \left\lfloor \log_{|\mathbb{B}|} \frac{d}{L}\right\rfloor\quad\text{and}\quad b_{n_0 + 1}^* = \cdots = b_{\ell + d - \ell}^* = \left\lceil \log_{|\mathbb{B}|} \frac{d}{L} \right\rceil \, .
\end{equation}
Then $b_1^*, b_2^*,\ldots, b_{d}^*$ is the optimal solution for \eqref{optimization}. 
Computing the optimal value, we complete the proof of Theorem~\ref{thm:lowerbound}.

\section{Numerical Results}\label{sec:numerical}

 \begin{figure}[!t]
	\centering
	\includegraphics[width=\textwidth]{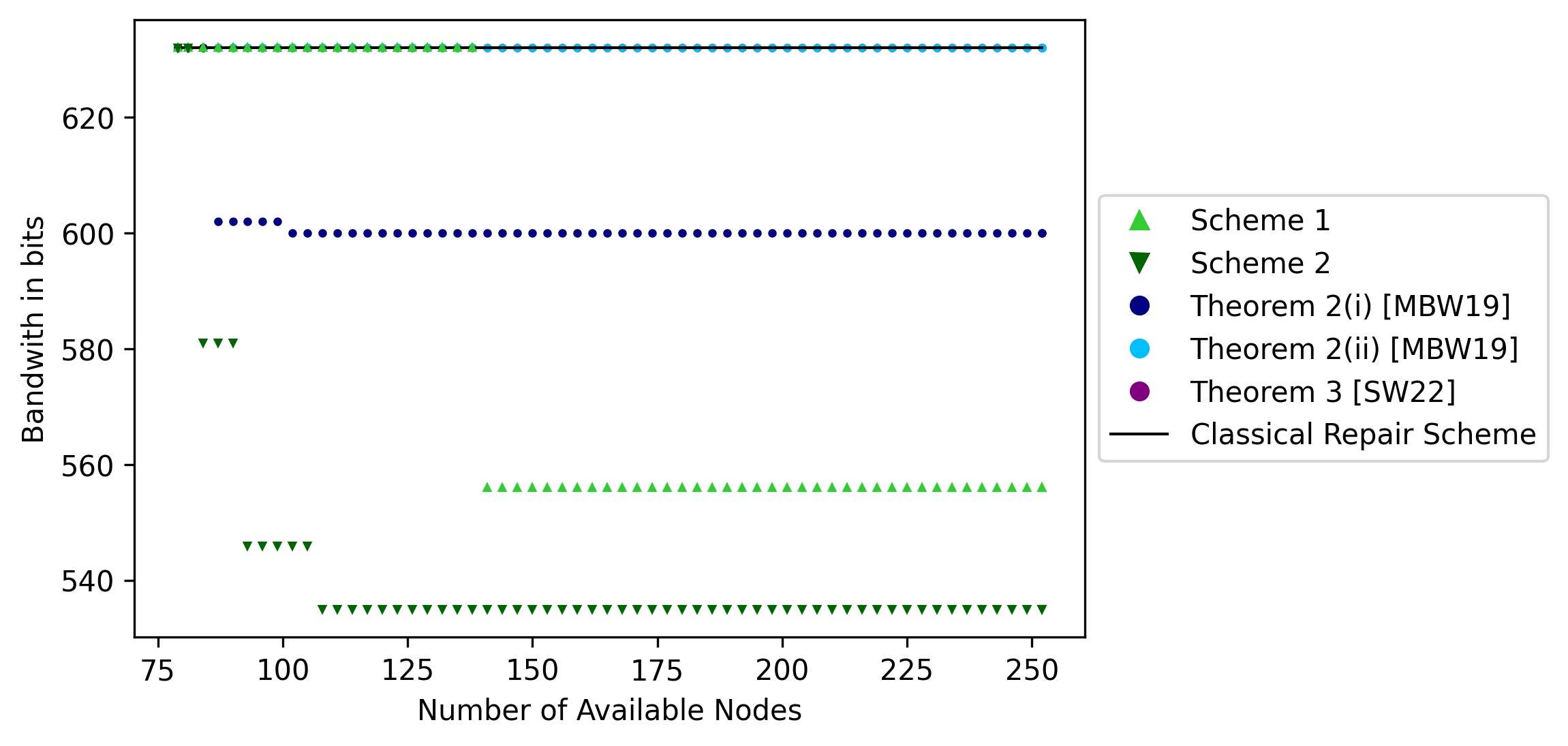}
	\caption{Recovery bandwidth for sums of $\ell=4$ coded symbols of RS code over $\textrm{GF}(2^{8})$ with $k=79$}
	\label{2}
\end{figure}

In this section, we compare the evaluation bandwidth of our schemes with those in previous work and lower bounds. Motivated by practical applications, we consider two different scenarios. In the first one, we vary the number of available nodes while the size $\ell$ of the weighted sum is fixed. In the second one, we change $\ell$ while the number of available nodes is such that it ensures the best bandwidth value.

\begin{enumerate}[(a)]
    \item {\bf Varying Number of Available Nodes, $d$}. Let us set $\ell=2$ and $k=79$. Consider $\mathbb{F}=\textrm{GF}(2^8)$ and $n=|\mathbb{F}|-\ell=254$ nodes. Hence, the classical evaluation scheme can compute any weighted sum $\SSS$ of two coded symbols using any $k=79$ coded symbols, or equivalently, downloading $632$ bits from any $79$ available nodes.
    
    First, suppose that $d=141$ nodes are available. Then Schemes~\ref{scheme:trace} and~\ref{scheme:subspace} are able to evaluate $\SSS$ with $436$ and $423$ bits, respectively. 
    In contrast, the schemes in Theorem~\ref{th::MBW}(i)~\cite{Mardia} and Theorem~\ref{th::MBW}(ii)~\cite{Mardia} have evaluation bandwidth $496$ and $558$ bits, respectively, while the schemes in Theorem~\ref{th::SW}~\cite{SW} is unable to improve the classical evaluation bandwidth. 
    Suppose further that all $d=254$ nodes are available. Then the evaluation bandwidth of both Schemes~\ref{scheme:trace} and~\ref{scheme:subspace} is $410$ bits. 
    In contrast, the schemes in Theorem~\ref{th::MBW}(i)~\cite{Mardia} and Theorem~\ref{th::MBW}(ii)~\cite{Mardia} have evaluation bandwidth $496$ and $409$ bits, respectively, while the schemes in Theorem~\ref{th::SW}~\cite{SW} cannot improve the classical evaluation bandwidth. We do note that such an improvement of scheme in  Theorem~\ref{th::MBW}(ii)~\cite{Mardia} over our Schemes~\ref{scheme:trace} and~\ref{scheme:subspace} can be obtained only for the fields of small size and high number of available nodes $d$. In Figure~\ref{1}, we complete the picture and provide the evaluation bandwidth for various schemes and lower bounds for $79\leq d\leq 254$. We do note that during plotting the graph hereinafter we take minimum overall bandwidth for a smaller or equal number of available nodes. We also emphasize the significant gap between constructions and lower bounds that grows with the increase of $d$.

	Next, we consider the case of $\ell=4$ while the other settings remain the same. 
	Again, the classical evaluation bandwidth remains at 632 bits.
	Suppose that all $d=252$ nodes are available. Then Schemes~\ref{scheme:trace} and~\ref{scheme:subspace} are able to evaluate $\SSS$ with $556$ and $535$ bits, respectively. In contrast, the scheme in Theorem~\ref{th::MBW}(i)~\cite{Mardia} has evaluation bandwidth $600$ bits, while the schemes in Theorem~\ref{th::MBW}(ii)~\cite{Mardia} and Theorem~\ref{th::SW}~\cite{SW} cannot improve the classical evaluation bandwidth. In Figure~\ref{2}, we complete the picture and provide the evaluation bandwidth for various schemes and lower bounds for $79\leq d\leq 252$.

    \item {\bf Varying Size of Weighted Sum, $\ell$}. Let us consider the case of $\textrm{GF}(2^{32})$ and $k=2^{30}$ when our Schemes~\ref{scheme:trace} and~\ref{scheme:subspace} and scheme from Theorem~\ref{th::SW}~\cite{SW} can provide results that differ from the classical evaluation scheme. The values of bandwidth in bits for different values of $\ell$ are presented in Table~\ref{table}. Note that we vary the number of code symbols participated in the recovery process to obtain the best bandwidth value, but such a number must be less than $q^t-\ell$ as $\ell$ symbols are assumed to be erased. We see that for high values of $\ell$, the scheme from Theorem~\ref{th::SW}~\cite{SW} provides better results. In contrast, for small values of $\ell$, results by Schemes~\ref{scheme:trace} and~\ref{scheme:subspace} are better. Another interesting observation is that the gap between constructions and lower bounds decreases with $\ell$. We do note that such behavior also holds for the higher size of field and large values of $k$ when our Schemes~\ref{scheme:trace} and~\ref{scheme:subspace} and the scheme from Theorem~\ref{th::SW}~\cite{SW} improves over the  classical evaluation scheme. 
    
    \begin{table}[!t]
    \centering
\begin{tabular}{|l|r|r|r|r|r|r|}
\hline
Scheme                                  & $\ell=k$        & $\ell=\frac{k}{2}$ 		& $\ell=\frac{k}{2^5}$ & $\ell=\frac{k}{2^{12}}$ & $\ell=\frac{k}{2^{20}}$ & $\ell=\frac{k}{2^{28}}$ \\ \hline
Scheme~\ref{scheme:trace}               & $32\cdot2^{30}$ & $32\cdot2^{30}$   		& $32\cdot2^{30}$      & $32\cdot2^{30}$         		& $\pmb{17\cdot2^{30}}$        & $8\cdot2^{30}$          \\ \hline
Scheme~\ref{scheme:subspace}            & $32\cdot2^{30}$ & $32\cdot2^{30}$   		& $31.9\cdot2^{30}$    		& $\pmb{25.5\cdot2^{30}}$   & $\pmb{17\cdot2^{30}}$         & $\pmb{7.5\cdot2^{30}}$        \\ \hline
Theorem~\ref{th::MBW}(i)~\cite{Mardia}  & $32\cdot2^{30}$ & $32\cdot2^{30}$   		& $32\cdot2^{30}$      		& $26.5\cdot2^{30}$         & $19\cdot2^{30}$         & $8.5\cdot2^{30}$         \\ \hline
Theorem~\ref{th::MBW}(ii)~\cite{Mardia} & $32\cdot2^{30}$ & $32\cdot2^{30}$   		& $32\cdot2^{30}$     		& $32\cdot2^{30}$         	& $32\cdot2^{30}$         & $12\cdot2^{30}$         \\ \hline
Theorem~\ref{th::SW}~\cite{SW}          & $32\cdot2^{30}$ & $\pmb{26\cdot2^{30}}$   &  $\pmb{26\cdot2^{30}}$     & $26\cdot2^{30}$         	& $26\cdot2^{30}$         & $26\cdot2^{30}$         \\ \hline
Classical Repair Scheme                 & $32\cdot2^{30}$ & $32\cdot2^{30}$   		& $32\cdot2^{30}$      & $32\cdot2^{30}$         		& $32\cdot2^{30}$         & $32\cdot2^{30}$         \\ \hline
Integer (Lower) Bound & $2$ & $1\cdot2^{30}$ & $1.9\cdot 2^{30}$ & $2\cdot 2^{30}$ & $2\cdot 2^{30}$ & $2\cdot 2^{30}$ \\ \hline
Fractional (Lower) Bound & $2$ & $1.6\cdot 2^{29}$ & $1.6\cdot 2^{30}$ & $1.7 \cdot 2^{30}$ & $1.7 \cdot 2^{30}$ & $1.7 \cdot 2^{30}$ \\ \hline

\end{tabular}
\caption{Recovery bandwidth for Reed-Solomon code over $\textrm{GF}(2^{32})$ and $k=2^{30}$. We measure the bandwidth in bits and we vary the number of code symbols $d$ that participated in the recovery process to obtain the best bandwidth value while satisfying theorems requirements.
}
\label{table}
\end{table}
    
\end{enumerate}

\section{Conclusion and Open Problems}
\label{sec:conclusion}
We investigate the problem of efficiently computing linear combinations of coded symbols. 
We characterize the evaluation process in terms of an evaluation matrix for any $\FF$-linear codes.
For Reed-Solomon codes, we propose explicit evaluation schemes as well as a lower bound on recovery bandwidth. 
%for Reed-Solomon code of arbitrary length and number of information symbols. 
Several open questions remain, including:
\begin{enumerate}[(A)]
	\item closing the gap between the lower bound on evaluation bandwidth and an explicit construction;
	\item developing low-bandwidth evaluation schemes for other classes of functions, for example, low-degree polynomials;
	\item extending the framework to other parameter regimes, for example, when $t$ is extremely large \cite{Barg} or over prime fields \cite{Tamo1};
	\item ensuring robustness, that is, ensuring correct evaluation in the presence of erroneous nodes.
\end{enumerate}

\appendix

\section{Proof of the Form of Theorem~\ref{th::SW}}\label{SWPROOF}
In this appendix, we discuss on why \cite[Theorem 19]{SW} can be written as Theorem~\ref{th::SW}. For the case of $|{\cal B}| = \ell$, $|{\cal A}| = d$, $n = |\BB|^t$ from \cite[Theorem 19]{SW}, we have four conditions to be fulfilled:
\begin{enumerate}[(a)]
\item $\delta\ge \gamma + \frac{1}{|\BB|}$,
\item $\epsilon \le 1-\frac{k}{|\BB|^t}$,
\item $n-d = |I| <\gamma n$, and
\item $(\epsilon - \delta)|\BB|$ is an integer.
\end{enumerate}
From conditions (a)-(c), we obtain
\begin{align}
    \epsilon - \delta \le 1-\frac{k}{|\BB|^t}-\gamma -\frac{1}{|\BB|} < 1-\frac{k}{|\BB|^t}-\left(1-\frac{d}{|\BB|^t}\right) -\frac{1}{|\BB|} = \frac{d-k}{|\BB|^t} - \frac{1}{|\BB|},
\end{align}
which implies
\begin{align}\label{eq:SWproof}
    (\epsilon - \delta)|\BB| < \frac{d-k}{|\BB|^{t-1}} -1.
\end{align}
Let $s = (\epsilon-\delta)|B|\in\mathbb{Z}$ (from condition (d)), then we can write \eqref{eq:SWproof} as
\begin{align}
    d \ge |\BB|^{t-1}(s+1)+k + 1.
\end{align}
Combining it with the fact that $d+\ell\leq n$, we translate \cite[Theorem 19]{SW} to Theorem~\ref{th::SW}.

\section{Characterization of Evaluation Schemes using Evaluation Matrices}

In this appendix, we complete the proof of Theorem~\ref{thm:characterization}.
In other words, we demonstrate the necessary and sufficient conditions for the existence of low-bandwidth evaluation schemes.

\noindent $[\Leftarrow]$ First, we prove the sufficient condition. 
That is, using an evaluation matrix, we construct the functions $g_{j}:\;\mathbb{F}\to \mathbb{B}^{b_{j}}$ for $j\in [d]$ 
such that $\sum_{j\in [t]} b_j = b$. 
Hence, we obtain an $(\pmb{\kappa},\C)$-weighted-sum evaluation scheme with bandwidth $b$. 

Now, using Condition (C1) in Definition~\ref{def:evalmatrix}, the $m$-th column $\pmb{M}[:,m]$ belongs to the dual code $\mathcal{C}^{\perp}$ for all $m\in[t]$. This means that, for all $\pmb{c}=(c_1,\ldots,c_d,c_{d+1},\ldots,c_{\ell+d})\in\mathcal{C}$, we have that
\begin{equation*}
	\sum_{i\in[\ell + d]}c_i \pmb{M}[i,m]=0\text{, or, }
	\sum_{i\in [\ell]}c_i \pmb{M}[i,m]=-\sum_{j\in[d]}c_{\ell+j} \pmb{M}[\ell+j,m].
\end{equation*}

Next, we multiply the equation by $\Lambda \triangleq \prod_{i\in [\ell]}\kappa_i$. Then using Condition~(C2), we have that 
$\Lambda \bM[i,m] = \kappa_i \bM_0[1,m]$, and so,
\begin{equation*}
	\sum_{i\in [\ell]}\kappa_i c_i \bM_0[m] = -\sum_{j\in[d]}\Lambda c_{\ell+j} \pmb{M}[\ell+j,m].
\end{equation*}

Applying trace function to both sides of the equation above and utilizing its linearity, we obtain
\begin{equation*}
	\tr\left(\left(\sum_{i\in [\ell]}\kappa_i c_i\right) \bM_0[m] \right)=-\tr\left(\sum_{j\in[d]}\Lambda c_{\ell+j} \bM[\ell+j,m]\right).
\end{equation*}

For $j\in[d]$, we set $b_{j} \triangleq \rank_{\mathbb{B}}(\pmb{M}[\ell+j,:])$ and form a $\BB$-basis $\chi_{1}^{j},\ldots,\chi_{b_{j}}^{j}$ corresponding to the column space of $\pmb{M}[\ell+j,:]$. 
Then for $m\in[t]$, we can write $\Lambda c_{\ell+j} \pmb{M}[\ell+j,l]$ as a $\BB$-linear combination of elements in $\chi_{1}^{j},\ldots,\chi_{b_{j}}^{j}$. So, for some $a_{j,1,m},\ldots,a_{j,b_{j},m}\in \BB$, we have that
\begin{equation}\label{eq:traceM}
	\tr\Big(\Lambda c_{\ell+j}\pmb{M}[\ell+j,m]\Big)=a_{j,1,m}\tr\left(c_{\ell+j}\Lambda\chi_{1}^{j}\right)+\cdots+a_{j,b_{j},m}\tr\left(c_{\ell+j}\Lambda\chi_{b_{j}}^{j}\right).
\end{equation}

Now, for $j\in [d]$, we define the function
\begin{equation}
	g_{j}(x):\; \mathbb{F}\to \mathbb{B}^{b_{j}},\; x\mapsto \left(\tr\left(x\Lambda\chi_{1}^{j}\right),\ldots,\tr\left(x\Lambda\chi_{b_{j}}^{j}\right)\right).
\end{equation}
Then we observe that for all $m\in[t]$, $\tr(\Lambda c_{\ell+j}\pmb{M}[\ell+j,l])$ can be computed using  $g_{j}(c_{\ell+j})$ and  $a_{j,1,m},\ldots,a_{j,b_{j},m}$. Then using \eqref{eq:traceM}, we can compute $\tr\left(\left(\sum_{i\in [\ell]}\kappa_ic_i\right)\pmb{M}_0[m]\right)$ for all  $m\in[t]$.

Finally, from Condition (C2), we have that the rank of the matrix $\pmb{M}_0$ is $t$. 
In other words, $\{\bM_0[m]: m\in [t]\}$ forms a $\BB$-basis.
Then Proposition~\ref{prop:trace} states there exists a trace-dual basis $\{\widetilde{\bM_0}[m]: m\in [t]\}$ and we have that
\begin{equation*}
	\sum_{i\in [\ell]}\kappa_ic_i=\sum_{m\in [t]} \tr\left(\left(\sum_{i\in [\ell]}\kappa_ic_i\right)\pmb{M}_0[m]\right)\widetilde{\pmb{M}_0}[m].
\end{equation*}
It remains to determine the evaluation bandwidth. 
This follows from the definition of $g_j$ and Condition (C3). 
Specifically, since $\rank_{\mathbb{B}}(\pmb{M}[\ell+j,:]) = b_j$ for $j\in [d]$, 
we have that the evaluation bandwidth is $\sum_{j=1}^db_j = b$.
\vspace{3mm}

\noindent $[\Rightarrow]$ Here, we prove the necessary condition. That is, we construct the evaluation matrix $\bM$ using the functions $g_{j}$. Let us choose a basis $\{u_1,\ldots,u_t\}$ of $\mathbb{F}$ over $\mathbb{B}$, its trace-dual basis $\{\tilde{u}_1,\ldots,\tilde{u}_t\}$ and write down $\sum_{i\in [\ell]}\kappa_ic_i$ as a $\mathbb{B}$-linear combination of elements of the latter one, namely
\begin{equation}\label{eq:recmat}
	\sum_{i\in [\ell]}\kappa_ic_i=\sum_{m=1}^t\zeta_m\tilde{u}_m,
\end{equation}
where $\zeta_m\in\mathbb{B}$ for all $m\in[t]$. The fact that the considered recovery scheme is $\mathbb{B}$-linear means that coefficients $\zeta_m$ are $\mathbb{B}$-linear combinations of replies $\{g_j(c_{\ell+j}):\;j\in[d]\}$. As a result, we have that 
\begin{equation*}
    \zeta_m=\sum_{j=1}^d\eta_j^mg_j(c_{\ell+j}),
\end{equation*}
where $\eta_j^m\in\mathbb{B}$.

As all functions $g_j$ are $\mathbb{B}$-linear maps from $\mathbb{F}$ to $\mathbb{B}^{b_j}$, we can define $g_j(c_{\ell+j})$ as follows
\begin{equation*}
    g_j(c_{\ell+j})=\gamma_1^j\textrm{Tr}(\theta_1^jc_{\ell+j})+\cdots+\gamma_{b_j}^j\textrm{Tr}(\theta_{b_j}^jc_{\ell+j}),
\end{equation*}
where $\{\theta_1^j,\ldots,\theta_{b_j}^j\}\in\mathbb{F}$ are linearly independent as vectors over $\mathbb{B}$ and $\{\gamma_1^j,\ldots,\gamma_{b_j}^j\}\in\mathbb{B}$. Consequently,
\begin{align*}
\zeta_m=\sum_{j=1}^d[\xi_{m,j,1}\textrm{Tr}(\theta_1^jc_{\ell+j})+\cdots+\xi_{m,j,b_j}\textrm{Tr}(\theta_{b_j}^jc_{\ell+j})],  \end{align*}
where $\xi_{m,j,l}=\eta_j^m\gamma_{l}$ and $\xi_{m,j,l}\in\mathbb{B}$ for all $l\in[t]$.
Let us multiply both sides of \eqref{eq:recmat} by $u_m$, where $m\in[t]$, and apply trace function from $\mathbb{F}$ to $\mathbb{B}$. As a result we have
\begin{equation*}
\textrm{Tr}\left(\sum_{i\in [\ell]}\kappa_ic_iu_m\right)=\zeta_1\textrm{Tr}(\tilde{u}_1u_m)+\cdots+\zeta_m\textrm{Tr}(\tilde{u}_mu_m)+\cdots+\zeta_t\textrm{Tr}(\tilde{u}_tu_m),
\end{equation*}
but bases $\{u_1,\ldots,u_t\}$ and $\{\tilde{u}_1,\ldots,\tilde{u}_t\}$ are trace-dual. Hence, 
\begin{equation*}
\textrm{Tr}\left(\sum_{i\in [\ell]}\kappa_ic_iu_m\right)=\zeta_m=\sum_{j=1}^d\left(\xi_{m,j,1}\textrm{Tr}(\theta_1^jc_{\ell+j})+\cdots+\xi_{m,j,b_j}\textrm{Tr}(\theta_{b_j}^jc_{\ell+j})\right).    
\end{equation*}
Employing the fact that $\{\xi_{m,j,1},\ldots,\xi_{m,j,b_j}\}\in\mathbb{B}$ we have
\begin{equation*}
    \textrm{Tr}\left(\sum_{i\in [\ell]}\kappa_ic_iu_m\right)=\sum_{j=1}^d\textrm{Tr}\left((\xi_{m,j,1}\theta_{1}^j+\cdots+\xi_{m,j,b_j}\theta_{b_j}^j)c_{\ell+j}\right).
\end{equation*}
We can define $\mu_{m,j}=\xi_{m,j,1}\theta_{1}^j+\cdots+\xi_{m,j,b_j}\theta_{b_j}^j$, where $\mu_{m,j}\in {\rm span}_{\mathbb{B}}\{\theta_1^j,\ldots,\theta_{b_j}^j\}$. Note that we have chosen $\theta_1^j,\ldots,\theta_{b_j}^j$ so that $\rank_{\mathbb{B}}\{\mu_{1,j},\ldots,\mu_{t,j}\}=b_j$ as otherwise the recovery scheme can be trivially improved by replacing $\{\theta_1^j,\ldots,\theta_{b_j}^j\}$ with basis of $\{\mu_{1,j},\ldots,\mu_{t,j}\}$ over $\mathbb{B}$ that contains less than $b_j$ elements. Consequently,
\begin{equation}\label{eq:recmat2}
\textrm{Tr}\left(\sum_{i\in [\ell]}\kappa_ic_iu_m\right)=\sum_{j=1}^d\textrm{Tr}(\mu_{m,j}c_{\ell+j}).
\end{equation}

Let $\{\delta_1,\ldots,\delta_t\}$ be a basis of $\mathbb{F}$ over $\mathbb{B}$ and $\{\tilde{\delta}_1,\ldots,\tilde{\delta}_{t}\}$ be its trace-dual basis. Due to the linearity of code $\mathcal{C}$ over $\mathbb{F}$, \eqref{eq:recmat} is also valid for codeword $(\delta_lc_1,\ldots,\delta_lc_{\ell+d})$. Hence for all $l\in[t]$ we have
\begin{equation*}
	\tr\left(\left[\sum_{i\in [\ell]}\kappa_ic_i\right]u_m\delta_l\right)=\sum_{j\in[d]}\tr(\mu_{mj}c_{\ell+j}\delta_l).
\end{equation*}
Multiplying both sides by $\tilde{\delta}_l$, we obtain
\begin{equation*}
	\tr\left(\left[\sum_{i\in [\ell]}\kappa_ic_i\right]u_m\delta_l\right)\tilde{\delta}_l=\sum_{j\in[d]}\tr(\mu_{mj}c_{\ell+j}\delta_l)\tilde{\delta}_l.
\end{equation*}
Taking sums over all $l\in[t]$, we have  
\begin{equation}\label{eq:recmat3}
	u_m\sum_{i\in[\ell]}\kappa_ic_i=\sum_{j\in[d]}\mu_{mj}c_{\ell+j}.
\end{equation}

According to~\eqref{eq:recmat3}, for all $m\in[t]$, $\left(u_m\kappa_1,\ldots,u_m\kappa_\ell,-\mu_{m(\ell+1)},\ldots,-\mu_{m(\ell + d)}\right)$ are codewords of the dual code and we can form the columns of recovery matrix from them. In such a case the latter has the following form
$$ \pmb{M} = \left[ \begin{matrix} \kappa_1u_1 & \ldots & \kappa_1u_t \\ \vdots & \ddots & \vdots \\ \kappa_{\ell}u_1 & \ldots & \kappa_{\ell}u_t \\ -\mu_{1,\ell+1} & \ldots & -\mu_{t,\ell+1} \\ \vdots & \ddots & \vdots \\ -\mu_{1,\ell+d} & \ldots & -\mu_{t,\ell+d} \end{matrix} \right]. $$

It is clear that $\frac{1}{\kappa_1}\pmb{M}[1,:]=\cdots=\frac{1}{\kappa_\ell}\pmb{M}[\ell,:]=[u_1,\ldots,u_t]$ and has full rank over $\mathbb{B}$. 
Also, for $j\in[d]$, we have that $\rank_{\mathbb{B}}\pmb{M}[\ell+j,:]=\rank_{\mathbb{B}}(-\mu_{1,\ell+j},\ldots-\mu_{t,\ell+j})=b_j$. Therefore, we have $\sum_{j\in[d]}\rank_{\mathbb{B}}\pmb{M}[\ell+j,:]=\sum_{j\in[d]}b_j=b$.

\bibliographystyle{alpha}
\bibliography{ref_list}

%\printbibliography

\end{document}

%% file: lowbandwidth.tex
\vspace{-2mm}

In this section, we provide two low-bandwidth schemes that evaluates the weighted sum of $\ell$ Reed-Solomon coded symbols. 
Throughout this section, we set $\mathcal{B}\triangleq\{\beta_1,\ldots,\beta_{\ell}\}$ and $\mathcal{A}\triangleq\{\alpha_1,\ldots,\alpha_d\}$ to be two disjoint subsets of distinct points in $\mathbb{F}$.
Hence, we have that $\ell+d\leq|\mathbb{F}|$.
We further choose $f(x)$ to be a polynomial over $\mathbb{F}$ with degree less than $k$.
Suppose we have $d\geq k$ nodes and for $j\in[d]$ we store in node $j$, the value $f(\alpha_j)$.

Next, we fix an $\ell$-tuple of coefficients $\pmb{\kappa}\in\FF_*^{\ell}$ and our task is to compute $\SSS=\sum_{i=1}^\ell\kappa_i f(\beta_i)$ by downloading as little symbols from the other $d$ nodes. 
Specifically, we provide an $(\kappa,\C)$-weighted-sum evaluation scheme where $\C$ is the Reed-Solomon code $\rs({\cal B}\cup {\cal A},k)$.
To this end, we consider a base field $\mathbb{B}$ such that $\mathbb{F}$ is a field extension of degree $t$ over $\BB$.
We let $\{u_1,\ldots,u_t\}$ be an $\BB$-basis of $\mathbb{F}$ and $\{\Tilde{u}_1,\ldots,\Tilde{u}_t\}$ be the corresponding dual basis. 
Our first scheme (Scheme~\ref{scheme:trace}) uses trace polynomials defined by \eqref{eq:trace} and it demonstrates our general framework.
Later on, we describe Scheme~\ref{scheme:subspace} which uses subspace polynomials in lieu of trace polynomials.

\begin{algorithm}[t]
\caption{: Scheme based on Trace Polynomials. {\bf Bandwidth}: $d\log_2{|\mathbb{B}|}$ bits}
\label{scheme:trace}
\begin{algorithmic}
	\REQUIRE $\A\triangleq \{\alpha_j : j \in [d]\}$, $\B\triangleq \{\beta_j : j \in [\ell]\}$, \textcolor{black}{values of polynomial $f$ of degree at most $k-1$ in points $\mathcal{A}$, coefficients $\pmb{\kappa}\in\FF_*^{\ell}$}
	\ENSURE $\SSS \triangleq \sum_{i=1}^\ell \kappa_i c_i$
	\vspace{1mm}
	
	\STATE{\hspace*{-3mm}{\em Pre-computation Phase}}
	\STATE{\textcolor{black}{Choose $\BB$-basis $\{u_1,\ldots,u_t\}$ of $\FF$ and its trace-dual basis $\{\tilde{u}_1,\ldots,\tilde{u}_t\}$}}
	\FOR{$i\in[t]$, $j\in[d]$}
		\STATE{pre-compute $\sigma_{i,j}\triangleq \tr(u_i(\alpha_j-\beta_1)\cdots(\alpha_j-\beta_{\ell}))$}
	\ENDFOR
	\vspace{1mm}
	
	\STATE{\hspace*{-3mm}{\em Download Phase}}
	\FOR{$j\in[d]$}
		\STATE{download from node $j$
			\begin{equation}\label{eq:download-1}
				\tau_j=\tr\left(\frac{g(\alpha_j)\lambda_{\ell+j} f(\alpha_j)}{(\alpha_j-\beta_1)\cdots(\alpha_j-\beta_{\ell})}\right)\in\mathbb{B}
			\end{equation} 
			where
			\begin{equation}\label{eq:gx}
				g(x) = \sum_{j=1}^\ell\frac{\kappa_j}{\lambda_j}\frac{\prod_{\textcolor{black}{m\in[\ell],m\ne j}}(x-\beta_{\textcolor{black}{m}})}{\prod_{\textcolor{black}{m\in[\ell],m\ne j}}(\beta_j-\beta_{\textcolor{black}{m}})},
			\end{equation}
			and $\lambda_j$ are multipliers of the dual GRS code defined by~\eqref{GRSM}.
			}
	\ENDFOR
	\vspace{1mm}
	
	\STATE{\hspace*{-3mm}{\em Evaluation Phase}}
	\FOR{$i\in[t]$}
	\STATE{compute trace of $\SSS$ with respect to the basis element $u_i$. That is, 
	\begin{equation}\label{eq:traceS}
		\tr(u_i\SSS)=-\sum_{j=1}^d\sigma_{i,j}\tau_j    
	\end{equation}}
	\ENDFOR
	\STATE{Finally, recover $\SSS$ using Proposition~\ref{prop:trace}. That is, set $\SSS=\sum_{i=1}^t\tr(u_i\SSS)\Tilde{u}_i$.}
\end{algorithmic}
\end{algorithm}

First, we provide a formal description of Scheme~\ref{scheme:trace} and demonstrate its correctness.
\vspace{-2mm}

\begin{proof}[Proof of Theorem~\ref{thm:scheme}(i)]
	First, it is straightforward verify that only one sub-symbol in $\BB$ are downloaded from each node. Hence, the evaluation bandwidth is $d$.
	
	Next, we show that Scheme~\ref{scheme:trace} is correct. 
	Crucially, we prove that \eqref{eq:traceS} holds for all $i\in[t]$. 
	Now, since $(f(\omega))_{\omega\in\B\cup\A}$ belongs to $\rs(\B\cup\A,k)$, we consider its dual code $\grs(\mathcal{B}\cup\mathcal{A},n-k, \pmb{\lambda})$ where $\pmb{\lambda}$ is defined by \eqref{GRSM}. Next, we use trace polynomials to form parity check polynomials. For $i\in [t]$,  let
	\begin{equation}
		r_i(x) \triangleq \frac{g(x)\tr\left(u_i\prod_{j\in [\ell]} (x-\beta_j)\right)}{\prod_{j\in [\ell]} (x-\beta_j)},
	\end{equation}
	where $g(x)$ and $\pmb{\lambda}$ are defined by \eqref{eq:gx} and \eqref{GRSM}, respectively.
	
	Observe that $g(\beta_j)={\kappa_j}/{\lambda_j}$ for all $j\in[\ell]$.
	Hence, $r_i(\beta_j)=\frac{\kappa_j}{\lambda_j}u_i$ for all $j\in[\ell]$.
	Furthermore, the polynomial $r_i(x)$ has degree $\ell-1+\ell|\mathbb{B}|^{t-1}-\ell=\ell|\mathbb{B}|^{t-1}-1$. 
	As $d\geq\ell|\mathbb{B}|^{t-1}-\ell+k$, the polynomial $r_i(x)$ is indeed a parity-check polynomial for $f(x)$. So, it follows from \eqref{PCE} that for $i\in[t]$,
	\begin{align}
		r_i(\beta_1)\lambda_1f(\beta_1)+\dots+r_i(\beta_\ell)\lambda_\ell f(\beta_\ell) 
		&=-\sum_{j\in[d]}r_i(\alpha_j)\lambda_{\ell + j}f(\alpha_j), \text{~~~or,} \notag\\ u_i\kappa_1f(\beta_1)+\dots+u_i\kappa_\ell f(\beta_\ell)&=
		-\sum_{j\in[d]}r_i(\alpha_j)\lambda_{\ell+j}f(\alpha_j) \label{eq:rxfx}
	\end{align}
	Applying the trace function to both sides of the equation and utilizing its linearity, we obtain
	\vspace{-3mm}
	\begin{align*}
		\tr(u_i\SSS) & = 
		\tr\left(u_i\sum_{j\in[\ell]} \kappa_j f(\beta_j)\right) \\
		& = -\sum_{j\in[d]}\tr(r_i(\alpha_j)\lambda_{\ell+j}f(\alpha_j)) \\
		& = - \sum_{j\in[d]}\tr\left(u_i(\alpha_j-\beta_1)\cdots(\alpha_j-\beta_{\ell})\right)\tr\left(\frac{g(\alpha_j)\lambda_{\ell+j} f(\alpha_j)}{(\alpha_j-\beta_1)\cdots(\alpha_j-\beta_{\ell})}\right)\\
		& =-\sum_{j=1}^d\sigma_{i,j}\tau_j,    
	\end{align*}
\vspace{-5mm}
	as required.
\end{proof}
%\textcolor{red}{\begin{proof}[Proof of Corollary~\ref{cor::asym}(i)]
%    We can write the condition in Scheme 1 as $t-1 \le \log_{|\BB|}\frac{n-k}{\ell}$, as a result, the recovery bandwidth of Scheme 1 is
%    \begin{align}\label{bw:scheme2}
%        \min_{\ell\le \ell'\le n-k} (n-\ell')\left(t-\left\lfloor\log_{|\BB|}\frac{n-k}{\ell'}\right\rfloor\right).
%    \end{align}
%    By noting that $t = \log_{|\BB|} n$, the recovery bandwidth in \eqref{bw:scheme2} is at most
%    \begin{align}
%        (n-\ell) \left(1+\log_{|\BB|} \frac{\ell}{1-\frac{k}{n}}\right),
%    \end{align}
%    which asymptotically approximates to $\log_{|\BB|} \frac{\ell}{1-\frac{k}{n}}$ subsymbols
%    per surviving node, which is smaller in comparison to $\log_{|\BB|} \frac{2\ell}{1-\frac{k}{n}}$, the asymptotic bound in \cite{Mardia} for the case $\ell \ll n^{\frac{k}{n}}$.
%\end{proof}}
\vspace{3mm}
Next, we present an instructive example of Scheme~\ref{scheme:trace}.
\begin{example}
Let $\mathbb{B} = \textrm{GF}(4) = \{0,1,b,1+b\}$ with $b^2 = b+1$, and
let $\mathbb{F} = \textrm{GF}(4^2) = \{0,1,\alpha,\ldots,\alpha^{14}\}$ with $\alpha^2 = \alpha + b$ and $\alpha^4=\alpha+1$.
So, $\FF$ is the field extension of $\mathbb{B}$ with degree $t = 2$ and basis $\{u_1,u_2\}=\{1,\alpha\}$.

Set $\mathcal{B}=\{0, 1\}$ and $\mathcal{A}=\{\alpha,\ldots,\alpha^{14}\}$ with $\ell=2$ and $d=14$.
We consider the Reed-Solomon code $\C=\rs(\B\cup\A,k)$ with $k=8$ and the polynomial $f(x)=x^7$. 
In what follows, utilizing Scheme~1, we demonstrate the evaluation of $\SSS=f(0)+f(1)$ by downloading $d=14$ sub-symbols (in $\BB$) from the nodes corresponding to $\A$. 
		
In the pre-computation phase, we compute the 28 values $\sigma_{i,j}=\tr(u_i)$ for $i\in[2]$ and $j\in[14]$.
Then from the node storing $f(x)$, we download the sub-symbol $\tr\left(\frac{f(x)}{x(x-1)}\right)\in \BB$.
Recall that in this case, we have that $\tr(x)=x+x^4$. 

In summary, we have the following pre-computed and downloaded values.
\begin{center}
\small
\begin{tabular}{|p{14mm}| r|r|r|r|r| r|r|r|r|r| r|r|r|r|}
\hline
Node $j$ & 1 & 2 & 3 & 4 & 5 & 6 & 7 & 8 & 9 & 10 & 11 & 12 & 13 & 14 \\ \hline
Evaluation point & 	$\alpha$ 	& $\alpha^2$    & $\alpha^3$ & $\alpha^4$ & $\alpha^5$  &	
         			$\alpha^6$ 	& $\alpha^7$    & $\alpha^8$ & $\alpha^9$ & $\alpha^{10}$ & 
         			$\alpha^{11}$  & $\alpha^{12}$ & $\alpha^{13}$ &$\alpha^{14}$ \\ \hline
$\sigma_{1,j}$ 	& $0$ & $0$ & $1$ & $0$ & $0$ 		& $1$ & $1$ & $0$ & $1$ & $0$ 		& $1$ & $1$ & $1$ & $1$\\     
$\sigma_{2,j}$ 	& $b$ & $1+b$ & $1+b$ & $b$ & $1$ 	& $0$ & $1$ & $1+b$ & $1$ & $1$ 	& $b$ & $b$ & $0$ & $1+b$\\ \hline
$\tau_j$     	& $1$ & $1$ & $1$ & $1$ & $0$		& $1$ & $1+b$ & $1$ & $1$ & $0$ 	& $b$ & $1$ & $1+b$ & $b$ \\ \hline
\end{tabular}
\end{center}

\begin{comment}
\begin{align*}
			\begin{tabular}{|c|c|c|c|c|}
				\hline
				Node $j$ & Evaluation point & $\sigma_{1,j}$ & $\sigma_{2,j}$ & $\tau_j$\\
				\hline
				$1$ & $\alpha$ & $0$ & $b$ & $1$\\
				$2$ & $\alpha^2$ & $0$ & $1+b$ & $1$\\
				$3$ & $\alpha^3$ & $1$ & $1+b$ & $1$\\
				$4$ & $\alpha^4$ & $0$ & $b$ & $1$\\
				$5$ & $\alpha^5$ & $0$ & $1$ & $0$\\
				$6$ & $\alpha^6$ & $1$ & $0$ & $1$\\
				$7$ & $\alpha^7$ & $1$ & $1$ & $1+b$\\
				$8$ & $\alpha^8$ & $0$ & $1+b$ & $1$\\
				$9$ & $\alpha^9$ & $1$ & $1$ & $1$\\
				$10$ & $\alpha^{10}$ & $0$ & $1$ & $0$\\
				$11$ & $\alpha^{11}$ & $1$ & $b$ & $b$\\
				$12$ & $\alpha^{12}$ & $1$ & $b$ & $1$\\
				$13$ & $\alpha^{13}$ & $1$ & $0$ & $1+b$\\
				$14$ & $\alpha^{14}$ & $1$ & $1+b$ & $b$\\
				\hline
			\end{tabular}
		\end{align*}
\end{comment}

Then in the evaluation phase, we use these values and \eqref{eq:traceS} to obtain two independent traces: 
\begin{equation*}
	\tr(u_1\SSS)= 0\quad\text{and}\quad \tr(u_2\SSS)= 1,
\end{equation*}
and using trace-dual basis $\{\widetilde{u}_1, \widetilde{u}_2\}=\{\alpha^4, 1\}$, we obtain our desired value:
\begin{equation*}
f(0) + f(1) = \SSS = \tr(u_1\SSS)\widetilde{u}_1 + \tr(u_2\SSS)\widetilde{u}_2 = 0\cdot \alpha^4 + 1 \cdot 1 =  1.
\end{equation*}
In this case, we downloaded $14$ target traces of $2$ bits (in total $28$ bits), which is less than $32$ bits (if we download $8$ available symbols of $4$ bits each).
\end{example}
	%\subsection{Scheme based on subspace polynomials}
Scheme~\ref{scheme:trace} utilize trace polynomials to form parity-check equations. 
In the next scheme, we replace these polynomials by subspace-polynomials and hence, obtain more flexibility in terms of possible values of $t$ and $|\BB|$ (see Section~\ref{sec:numerical}).

\begin{definition}\label{SSP}
	Let $W$ be a $\BB$-subspace of $\FF$  with dimension $s<t$. 
	The {\em subspace polynomial defined by $W$} can be written as 
	\begin{equation}
		L_W(x)=\prod_{\alpha\in W}(x-\alpha)=x\prod_{\alpha\in W\setminus\{0\}}(x-\alpha)=\sum_{j=0}^se_jx^{|\mathbb{B}|^j}\,.
	\end{equation}
	%where $c_1=w_1\cdot...\cdot w_{|B|^s-1}\ne 0$.
\end{definition}

Observe that when $W$ is the kernel of the trace function~\eqref{eq:trace}, the subspace polynomial corresponding to $W$ is in fact the trace polynomial.
%Note that trace-polynomial is a special case of subspace polynomial whose set of values lie in a subfield of $\mathbb{F}$. Let us provide the necessary definition below.
Moreover, one can show that $L_W:\FF\to\FF$ is $\mathbb{B}$-linear mapping whose kernel is given by the subspace $W$. 
Therefore, the image of $L_W$ is a subspace of dimension $t-s$ over $\mathbb{B}$ and 
we let $\{\chi_1,\dots,\chi_{t-s}\}$ be a basis. 
Hence, each element of $\textrm{Im}(L_W)$ can be represented as a $\BB$-linear combination of the elements in $\{\chi_1,\dots,\chi_{t-s}\}$.

We are now ready to describe Scheme~\ref{scheme:subspace} and provide the proof of Theorem~\ref{thm:scheme}(ii).

\begin{proof}[Proof of Theorem~\ref{thm:scheme}(ii)]
	First, it is straightforward to verify that $(t-s)$ sub-symbols in $\BB$ are downloaded from each node. Hence, the evaluation bandwidth is $d(t-s)$.
	
	Next, we show that Scheme~\ref{scheme:subspace} is correct, i.e., we prove that~\eqref{eq:traceS(ii)} holds for all $i\in[t]$. Similarly, since $(f(\omega))_{\omega\in\B\cup\A}$ belongs to $\rs(\B\cup\A,k)$, we consider its dual code $\grs(\mathcal{B}\cup\mathcal{A},n-k, \pmb{\lambda})$ where $\pmb{\lambda}$ is defined by \eqref{GRSM}.  However, instead of using trace polynomials, we use subspace polynomials to form parity check polynomials. Specifically, for all $i\in[t]$, we set
	\begin{equation}
		r_{W,i}(x)\triangleq  \frac{g(x)L_{W}\left(u_i\prod_{j\in [\ell]} (x-\beta_j)\right)}{\prod_{j\in [\ell]} (x-\beta_j)},
	\end{equation}
	where $g(x)$ and $\pmb{\lambda}$ are defined by \eqref{eq:gx} and \eqref{GRSM}, respectively.
	
	Observe that $g(\beta_j)=\frac{\kappa_j}{\textcolor{black}{e}_0\lambda_j}$ for all $j\in[\ell]$. Hence, $r_{W,i}(\beta_j)=g(\beta_j)\textcolor{black}{e}_0u_i=u_i\frac{\kappa_j}{\lambda_j}$ for all $j\in[\ell]$. Furthermore, the polynomial $r_{W,i}(x)$ has degree $\ell-1+\ell|\mathbb{B}|^{s}-\ell=\ell|\mathbb{B}|^{s}-1$. As $d\geq \ell|\mathbb{B}|^s-\ell+k$, the polynomial $r_{W,i}(x)$ is indeed a parity-check polynomial for $f(x)$. So, it follows from~\eqref{PCE} that for $i\in[t]$,
	\begin{align}
	r_{W,i}(\beta_1)\lambda_1f(\beta_1)+\dots+r_{W,i}(\beta_\ell)\lambda_{\ell}f(\beta_\ell)&=-\sum_{j\in[d]}r_{W,i}(\alpha_j)\lambda_{\ell + j}f(\alpha_j), \text{~~~or,}\\
	\label{eq:pceforsubspace}
	u_i\kappa_1f(\beta_1)+\dots+ u_i\kappa_\ell f(\beta_\ell)&=-\sum_{j\in[d]}r_{W,i}(\alpha_j)\lambda_{\ell + j}f(\alpha_j)
	\end{align}
	Note that, for any $i\in [t]$,
	\begin{align*}
		r_{W,i}(\alpha_j) &= \frac{g(\alpha_j)}{(\alpha_j-\beta_1)\cdots(\alpha_j-\beta_l)}L_W(u_i(\alpha_j-\beta_1)\cdots(\alpha_j-\beta_\ell))\\
		&=\frac{g(\alpha_j)}{(\alpha_j-\beta_1)\cdots(\alpha_j-\beta_\ell)}(\sigma_{i,j,1}\chi_1+\cdots+\sigma_{i,j,t-s}\chi_{t-s}),
	\end{align*}
	where $\sigma_{i,j,1},\ldots,\sigma_{i,j,t-s}\in \mathbb{B}$ for all $j\in[d]$. Applying the trace function to both sides of the equation~\eqref{eq:pceforsubspace} and utilizing its linearity, we obtain
	\begin{align*}
		\tr(u_i\SSS) &= \tr\left(u_i\sum_{j\in[\ell]} \kappa_j f(\beta_j)\right)\\
		&= -\sum_{j\in [d]} \tr\left(r_{W,i}(\alpha_j) \lambda_{\ell+j} f(\alpha_j)\right)\\
		&= - \sum_{j\in[d]}\tr\left(\frac{\sigma_{i,j,1}g(\alpha_j)\lambda_{\ell+j}f(\alpha_j)\chi_1}{(\alpha_j-\beta_1)\cdot\ldots\cdot(\alpha_j-\beta_\ell)}+\cdots+\frac{\sigma_{i,j,t-s}g(\alpha_j)\lambda_{\ell+j}f(\alpha_j)\chi_{t-s}}{(\alpha_j-\beta_1)\cdot\ldots\cdot(\alpha_j-\beta_\ell)}\right)\\
		&= - \sum_{j\in[d]}\left(\sigma_{i,j,1}\tau_{j,1}+\cdots+\sigma_{i,j,t-s}\tau_{j,t-s}\right),
	\end{align*}
	as required.
\end{proof}
\begin{algorithm}[!t]
	\caption{: Scheme based on Subspace Polynomials. {\bf Bandwidth}: $d(t-s)\log_2{|\mathbb{B}|}$ bits}
	\label{scheme:subspace}
	\begin{algorithmic}
		\REQUIRE $\A\triangleq \{\alpha_j : j \in [d]\}$, $\B\triangleq \{\beta_j : j \in [\ell]\}$, \textcolor{black}{values of polynomial $f$ of degree at most $k-1$ in points $\mathcal{A}$, coefficients $\pmb{\kappa}\in\FF_*^{\ell}$, $\BB$-linear subspace $W$ of $\FF$ with dimension $s$.}
		\ENSURE $\SSS \triangleq \sum_{i=1}^\ell \kappa_i c_i$
		\vspace{1mm}
		
		\STATE{\hspace*{-3mm}{\em Pre-computation Phase}}
		\STATE{\textcolor{black}{Choose $\BB$-basis $\{u_1,\ldots,u_t\}$ of $\FF$, its trace-dual basis $\{\tilde{u}_1,\ldots,\tilde{u}_t\}$ and basis $\{\chi_1,\ldots,\chi_{t-s}\}$ of $\textrm{Im}(L_W)$}}
		\FOR{$i\in[t]$, $j\in[d]$}
		\STATE{pre-compute $\sigma_{i,j,1},\ldots,\sigma_{i,j,t-s}$ coefficients of $L_W(u_i(\alpha_j-\beta_1)\cdots(\alpha_j-\beta_\ell))$ in basis $\{\chi_1,\ldots,\chi_{t-s}\}$}
		\ENDFOR
		\vspace{1mm}
		
		\STATE{\hspace*{-3mm}{\em Download Phase}}
		\FOR{$j\in[d]$}
		\STATE{download from node $j$
			%\begin{align*}
			%\tau_{j,1}=\tr\left(\frac{g(\alpha_j)\lambda_{\ell+j} f(\alpha_j)\chi_1}{(\alpha_j-\beta_1)\cdots(\alpha_j-\beta_{\ell})}\right)\in \BB
			%\end{align*}
			%\vspace{-8mm}
			%\begin{align*}
			%\vdots
			%\end{align*}
			%\vspace{-8mm}
			\begin{equation}\label{eq:download-2}
				\tau_{j,m}=\tr\left(\frac{g(\alpha_j)\lambda_{\ell+j} f(\alpha_j)\chi_{t-s}}{(\alpha_j-\beta_1)\cdots(\alpha_j-\beta_{\ell})}\right)\in \BB \text{~~~~for $m \in [t-s]$,}
			\end{equation}
			where,
			\begin{align}
				g(x) = \sum_{j=1}^\ell\frac{\kappa_j}{\textcolor{black}{e_0}\lambda_j}\frac{\prod_{\textcolor{black}{m\in[\ell],m\ne j}}(x-\beta_k)}{\prod_{\textcolor{black}{m\in[\ell],m\ne j}}(\beta_j-\beta_k)},
			\end{align}
			and $\lambda_j$ are multipliers of the dual GRS code defined by~\eqref{GRSM}.
		}
		\ENDFOR
		\vspace{1mm}
		
		\STATE{\hspace*{-3mm}{\em Evaluation Phase}}
		\FOR{$i\in[t]$}
		\STATE{compute trace of $\SSS$ with respect to the basis element $u_i$. That is, 
			\begin{equation}\label{eq:traceS(ii)}
				\tr(u_i\SSS)=-\sum_{j=1}^d\big(\sigma_{i,j,1}\tau_{j,1}+\cdots+\sigma_{i,j,t-s}\tau_{j,t-s}\big)  
		\end{equation}}
		\ENDFOR
		\STATE{Finally, recover $\SSS$ using Proposition~\ref{prop:trace}. That is, set $\SSS=\sum_{i=1}^t\tr(u_i\SSS)\Tilde{u}_i$.}
	\end{algorithmic}
\end{algorithm}